\documentclass{llncs}

\usepackage{makeidx}  
\usepackage{times}
\usepackage{amsmath, amssymb}
\usepackage{graphicx}
\usepackage{epsfig}
\usepackage{color}
\usepackage{caption}
\usepackage{subfigure}
\usepackage{algorithm}
\usepackage{algorithmic}
\usepackage{fullpage}
\usepackage{xcolor}					
\usepackage{tikz}
\usetikzlibrary{arrows,shapes,snakes}

\newtheorem{observation}{Observation}
\newcounter{claimcounter}
\numberwithin{claimcounter}{theorem}
\newenvironment{myclaim}{\refstepcounter{claimcounter}{\par\noindent{\bf Claim \theclaimcounter:}}}{\par}
\newtheorem{innercustomcase}{Case}
\newenvironment{customcase}[1]
	{\renewcommand\theinnercustomcase{#1.}\innercustomcase}
	{\endinnercustomcase}

\newcommand{\scr}[1]{\mathcal{#1}}
\newcommand{\textscr}[1]{\mathcal{#1}}

\begin{document}

\title{LP Rounding and Combinatorial Algorithms for Minimizing Active and
Busy Time
\thanks{This work has been supported by NSF Grants
CCF-1217890 and CCF-0937865.  A preliminary version of this
paper appeared in ACM Symposium on Parallelism in Algorithms
and Architectures (SPAA 2014).}}

\author{Jessica Chang \and
	Samir Khuller \and
	Koyel Mukherjee}
\institute{University of Maryland, College Park \newline
\email{\{jschang,samir,koyelm\}@cs.umd.edu}}
\maketitle

\begin{abstract}
We consider fundamental scheduling problems motivated by 
energy issues. In this framework, we are given a set of jobs, each with
a release time, deadline and required processing length. 
The jobs need to be scheduled on a machine 
so that at most $g$ jobs are active at any given time.
The duration for which a machine is active (i.e., ``on'') 
is referred to as its \emph{active time}. 
The goal is to find a feasible schedule for all jobs, 
minimizing the total active time.
When preemption is allowed at integer time points, 
we show that a minimal feasible schedule already
yields a 3-approximation (and this bound is tight)
and we further improve this to a 2-approximation via
LP rounding techniques.
Our second contribution is for the non-preemptive
version of this problem. However, since even asking if 
a feasible schedule on one machine exists is NP-hard, we
allow for an unbounded number of virtual machines, each
having capacity of $g$. This problem is known as the 
\emph{busy time} problem in the literature and a 4-approximation
is known for this problem. We develop a new combinatorial algorithm
that gives a $3$-approximation. 
Furthermore, we consider the preemptive busy time problem,
giving a simple and exact greedy algorithm 
when unbounded parallelism is allowed, i.e., $g$ is unbounded.
For arbitrary $g$, this yields an algorithm that is $2$-approximate.
\end{abstract}


\section{Introduction}
\label{sec:intro}

Scheduling jobs on multiple parallel or batch machines has received
extensive attention in the computer science and 
operations research communities for decades.  
For the most part, these studies have focused primarily 
on ``job-related'' metrics such as minimizing makespan, 
total completion time, flow time, tardiness
and maximizing throughput under various deadline constraints.
Despite this rich history, some of the most environmentally
(not to mention, financially) costly scheduling
problems are those driven by a pressing need to
reduce energy consumption and power costs, e.g. at data centers.
In general, this need is not addressed by the traditional 
scheduling objectives.  Toward that end, our work is most 
concerned with minimization of the total time that a machine 
is on \cite{Chang12,IPDPS09,FSTTCS10,mertzios2012} 
to schedule a collection of jobs. 
This measure was recently introduced in an effort to 
understand energy-related 
problems in cloud computing contexts, and the 
busy/active time models cleanly 
capture many central issues in this space.  
Furthermore, it has connections to several key problems
in optical network design, perhaps most notably
in the minimization of the fiber costs of Optical 
Add Drop Multiplexers (OADMs) \cite{IPDPS09}.  
The application of busy time models to optical network design 
has been extensively outlined in the literature 
\cite{IPDPS09,EuropAr08,ISAAC05,WZ03}. 

With the widespread adoption of data centers and cloud computing,
recent progress in virtualization has facilitated the consolidation of 
multiple virtual machines (VMs) into fewer hosts. 
As a consequence, many computers can be shut off, 
resulting in substantial power savings.
Today, products such as Citrix XenServer an VMware Distributed Resource
Scheduler (DRS) offer VM consolidation as a feature. 
In this sense, minimizing busy time (described next) 
is closely related to the basic problem of mapping VMs to physical hosts.

In the active time model, 
the input is a set $\scr{J}$ of $n$ jobs $J_1, \ldots, J_n$ that
needs to be scheduled on one machine.  Each job $J_j$ has
release time $r_j$, deadline $d_j$ and length $p_j$.
We assume that time is \textit{slotted} (to be defined formally)
and that all job parameters are integral.
At each time slot, we have to decide if 
the machine is ``on'' or ``off''. When the machine
is on at time $[t, t+1)$, time slot $t$ is active, 
and we may schedule one unit of up to $g$ distinct jobs in it,
as long as we satisfy the release time and deadline 
constraints for those jobs. When the machine is off at time $[t, t+1)$,
no jobs may be scheduled on it at that time.
Each job $J_j$ has to be scheduled in at least $p_j$ active slots between $r_j$ and $d_j$.  The goal is to schedule all the jobs while minimizing 
the {\em active time} of the machine, 
i.e., the total duration that the machine is on.
%
In the special case where jobs have unit
length ($p_j=1$ for all $j$), there is a fast exact algorithm due to 
Chang, Gabow and Khuller~\cite{Chang12}.
However, in general, the exact complexity of 
minimizing active time remains open. In this work, 
we demonstrate that any \textit{minimal feasible solution}
has active time within 3 of the optimal active time.
We also show that this bound is tight. We then further improve the
approximation ratio by considering a natural IP formulation and
rounding a solution of its LP relaxation to obtain 
a solution that is within twice the integer optimum 
(again this bound is tight). 
We note that if the on/off decisions are given, 
then determining feasibility is straightforward by reducing the
problem to a flow computation.  
We conjecture that the active time problem itself is likely NP-hard.
%

We next consider a slight variant of the active time problem, called the
{\em busy time} problem. 
Again, the input is a set $\scr{J}$ of $n$ jobs $J_1, \ldots J_n$ that
need to be scheduled.  Each job $J_j$ has release time 
$r_j$, deadline $d_j$ and length $p_j$.  In the busy time problem,
we want to assign and non-preemptively schedule the jobs 
over a set of identical machines.  (A job is non-preemptive means
that once it starts, it must continue to be processed without interruption
until it completes.  In other words, if job $J_j$ starts at time
$s_j$, then it finishes at time $s_j + p_j$.)
We want to partition jobs into groups so that at most $g$ jobs 
are running simultaneously on a given machine. 
Each group will be scheduled on its own machine.

A machine is {\em busy} at time $t$ means that 
there is at least one job running on it at $t$; 
otherwise the machine is {\em idle} at $t$. The amount of time 
during which a machine $M$ is busy is called its 
{\em busy time}, denoted {\em busy(M)}. 
The objective is to find a feasible 
schedule of all the jobs on the machines (partitioning jobs into
groups) to minimize the cumulative busy time over all the 
machines.  We call this the \textit{busy time problem}, consistent
with the literature~\cite{IPDPS09,FSTTCS10}.
That the schedule has access to an unbounded number of machines 
is motivated by the case when each group is a virtual machine.

A well-studied special case of this model is one in which each job $J_j$ 
is \textit{rigid}, i.e.  $d_j = p_j + r_j$.  
Here, there is no question about when each job must start.
Jobs of this form are called {\em interval jobs}. 
(Jobs that are not interval jobs are called \textit{flexible} jobs.)
The busy time problem for interval jobs is $NP$-hard~\cite{WZ03} 
even when $g=2$. 
Thus, we will look for approximation 
algorithms. What makes this special case 
particularly central is that one can convert an instance 
of flexible jobs to an instance of interval jobs in polynomial time,
by solving a dynamic program with unbounded $g$ \cite{FSTTCS10}.
The dynamic program's solution fixes the positions of the jobs to minimize
their ``shadow'' (projection on the time-axis, formally defined
in Section 4). The shadow of this
solution with $g=\infty$ is the smallest possible of
any solution to the original problem and can lower bound
the optimal solution for bounded $g$.  Then, we adjust the release
times and deadlines to artificially fix the position of each job
to where it was scheduled in the solution for unbounded $g$.
This creates an instance of interval jobs. 
We then run an approximation for interval jobs on this
instance.  Figure~\ref{fig1} shows a set of jobs 
and the corresponding packing
that yields an optimal solution, i.e., minimizing busy time.

\begin{figure}[htbp]
	\centering
	\includegraphics[width=0.75\textwidth]{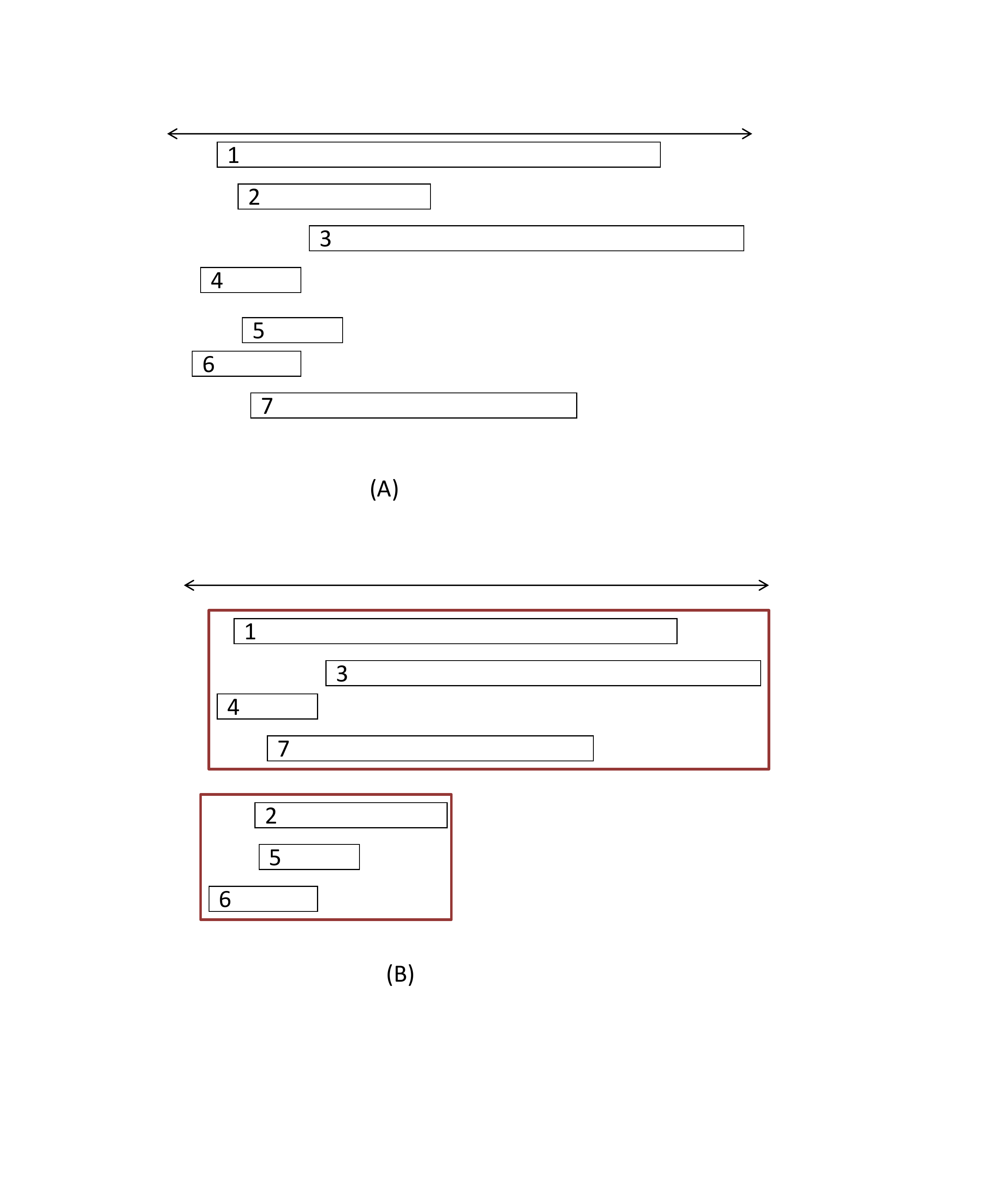}
	\caption{(A) Collection of interval jobs with unit demand, 
numbered arbitrarily. (B) Optimal packing of the jobs on two 
machines with $g=3$ minimizing total busy time.} 
\label{fig1}
\vspace{-10pt}
\end{figure}

Busy time scheduling in this form was first studied by 
Flammini et al. \cite{IPDPS09}. They prove that a simple 
greedy algorithm \textsc{FirstFit} for interval jobs is
$4$-approximate.
It considers jobs 
in non-increasing order by length, greedily packing each 
job in the first group in which it fits.  
In the same paper, they highlight an instance on which the cost of  
\textsc{FirstFit}  is three times that of the optimal solution.
Closing this gap would be very interesting\footnote{In 
an attempt to improve approximation guarantees, 
Flammini et al.~\cite{IPDPS09} consider two special cases.
The first case pertains to ``proper intervals'', where no job's 
interval is strictly contained in that of another.
For instances of this type, they show that the greedy algorithm
ordering jobs by release times is actually 2-approximate.
The second special case involve instances whose corresponding interval
graph is a clique - in other words, there exists a time $t$
such that each interval $[r_j,d_j)$ contains $t$.
In this case, a greedy algorithm also yields a 2-approximation.
As with proper intervals, it is not obvious that 
minimizing busy time on clique instances is NP-hard.
However, when the interval jobs are both proper 
{\em and} form a clique, a very simple dynamic 
program gives an optimal solution \cite{mertzios2012}. }. 

However, unknown to Flammini et al, earlier work by Alicherry and Bhatia
\cite{ESA03} and Kumar and Rudra \cite{FSTTCS05} already considered
a problem in the context of wavelength assignment. Their algorithms
immediately yield two different 2-approximations for 
the busy time problem with interval jobs (see the Appendix).

Khandekar et al. \cite{FSTTCS10} consider the generalization 
in which each job has an associated width or ``demand'' on
its machine.
For any set of jobs assigned to the same machine, the cumulative
demand of the active ones can be at most $g$ at any time.  
The authors apply \textsc{FirstFit} ideas to this problem
and obtain a 5-approximation.
The main idea involves partitioning 
jobs into those of ``narrow'' and ``wide'' demand.
Each wide job is assigned to its own machine, while
\textsc{FirstFit} is applied to the set of narrow jobs.
In addition, the authors give improved bounds for
special cases of busy-time scheduling with jobs of unit
demand.  When the interval jobs form a clique, they provide
a PTAS.  They also give an
exact algorithm when the intervals of the jobs are laminar, i.e.
two jobs' intervals intersect only if one interval is contained
in the other. However, we note that for the case of unit width
jobs, the same approach gives a 4-approximation for flexible
jobs, by solving a dynamic program for $g=\infty$. 
It turns out that the methods of Kumar and Rudra~\cite{FSTTCS05}
and Alicherry and Bhatia~\cite{ESA03} can be similarly
extended to also give 4-approximations.  There are
examples demonstrating that this analysis is tight.
We break this barrier with a different approach, obtaining
a 3-approximation.

\subsection{Problem Definition}
In this section we formally define the notions of \emph{active time} and 
\emph{busy time}. Both models are motivated 
by the total amount of time that a machine is actively working. 

\subsubsection{Active Time.}
The input consists of a set of jobs $\textscr{J}$, where each job $j$ 
has a release time $r_j$, a deadline $d_j$, and a length $p_j$. 
We let slot $t$ denote the unit of time $[t-1,t)$. 
Since time is slotted,
job $j$ can start as early at $r_j$ and as late as $d_j -p_j$.
For example, if a unit-length job has release time 1 and deadline 2,
it can be scheduled in slot $t=2$, but not in slot $t=1$.  
Equivalently, it can have a start time of 1, but not a start
time of zero.
The set of slots $\{r_j+1,\ldots,d_j\}$ comprise job $j$'s
\textit{window}.
We sometimes abuse notation and let 
$[r_j,d_j)$ refer to the slots in $j$'s window.
Then $p_j$ units of job $j$ must be scheduled in its window,
but not necessarily in consecutive slots (in other words
jobs can be considered to be a \emph{chain} of  
$p_j$ unit jobs), with identical release times and deadlines
and the restriction that in any time slot, 
at most one of these unit jobs can be scheduled. 
The running times of the algorithms are  polynomial in
$n$ and $P (=\sum_{j \in \textscr{J}}{p_j})$. 
We have access to a single machine
that is either active (`on') or not at any point of time.
The machine can process only $g$ jobs at any time instant. 
Since there is a single machine, we simply refer to the time 
axis henceforth in place of the machine. 
Assume without loss of generality that the
earliest release time of any job $j \in \textscr{J}$ is $0$ and
denote by $T$ 
the latest relevant time slot, i.e., $T = \max_j d_j$.
Then, it will be convenient to let $\textscr{T}$ refer to
the set of time slots $\{1,\ldots,T\}$.

When preemption is not allowed, determining whether
there exists a feasible solution for
non-unit length jobs becomes strongly NP-hard, by a reduction
from 3-PARTITION, even for the special case
when the windows of all the jobs are identical.

\subsubsection{Busy Time.}
Given that even determining the feasibility for
the non-preemptive problem is
hard in the active time model, we consider a relaxation of the model.
The key difference between busy time and active time is that while
active time assumes access to a single machine,
busy time can open an unbounded number of machines if necessary.
(One can think of each machine as a virtual machine.)
As in the active time problem, there is a set of jobs,
$\textscr{J}$, where each job $j$
has a release time $r_j$, a deadline $d_j$,
and a length $p_j$, and each machine has capacity $g$.
There is no restriction
on the integrality of the release times or deadlines.
The jobs need to be partitioned into groups so that when
each group is scheduled non-preemptively on its own machine,
at most $g$ jobs are running simultaneously on a given machine. 
We say that a machine is {\em busy}
at time $t$ if there is at least one job running on the machine
at $t$; otherwise the machine is {\em idle}. The time intervals
during which a machine $M$ is processing at least one job is called its 
{\em busy time}, denoted as {\em busy(M)}. 
The goal is to partition the jobs onto machines so that no
machine is working on more than $g$ jobs at a time, and
the cumulative busy time over all machines is minimized.
We will call this the \textit{busy time problem}.
Note that every instance is feasible in the busy time model.

\subsection{Our Results}
For the active time problem when we are allowed preemption at
integer time points and time is slotted, we first show that considering 
any minimal feasible solution gives us a 3-approximation.  A minimal feasible
solution can be found by starting with a feasible solution and making slots 
inactive in any order, as long as the instance remains feasible (we will explain later
how to test feasibility given a set of active slots).
We then consider
a natural IP formulation for this problem and show that considering
an LP relaxation allows us to convert a fractional schedule to an
integral schedule by paying a factor of 2. As a by product, this 
yields a 2-approximation. We note that the integrality gap
of 2 is tight~\cite{Chang12}.
 
%
Since the busy time problem for interval jobs 
is NP-hard \cite{WZ03}, the focus in this paper is
the development of a polynomial-time algorithm \textsc{GreedyTracking}
with a worst-case approximation guarantee of 3, improving the
previous bounds of 4 (as mentioned earlier, there seem to be
several different routes to arrive at this bound). 
As before we use the dynamic program to first solve 
the problem for unbounded $g$ \cite{FSTTCS10}
and then reduce the problem to the case of interval jobs.
The central idea 
is to iteratively identify a set of disjoint jobs;
we call such a set a ``track''.  Then, the subset of jobs assigned to
a particular machine is the union of $g$ such tracks; we call 
the set of jobs assigned to the same machine a \textit{bundle} of jobs.
Intuitively, this approach is less myopic than \textsc{FirstFit},
which schedules jobs one at a time.
We also construct examples where \textsc{GreedyTracking}
yields a solution twice that of the optimum for interval jobs.

One important consequence of \textsc{GreedyTracking} is an
improved bound for the busy time problem on flexible jobs.
Similar to Khandekar et al.~\cite{FSTTCS10}, 
we first solve the problem assuming
unbounded machine capacity to get a solution that minimizes the
projection of the jobs onto the time-axis.  Then, we 
can map the original instance to one of interval jobs,
forcing each job to be done exactly as it was in the
unbounded capacity solution.  We prove that in total,
this approach has busy time within 3 times that of 
the optimal solution.  In addition, we explore the 
preemptive version of the problem and provide a greedy 
$2$-approximation.

\subsection{Related Work}
While in both active time and busy time models, 
we assign jobs to machines where up to $g$ jobs can run concurrently,
the key difference between the two models is that
the former model operates on a single machine,
while the latter assumes access to an unbounded number of machines.
In the active time model, when jobs are unit in length, 
Chang, Gabow and Khuller~\cite{Chang12} present a fast
linear time greedy algorithm.  When
the release times and deadlines can be real numbers, they
give an $O(n^7)$ dynamic program to solve it; this result
has since been improved to an $O(n^3)$-time algorithm
in the work of Koehler and Khuller~\cite{KK13}.
In fact, their result holds even for a finite
number of machines.  Chang, Gabow and Khuller~\cite{Chang12} 
also consider generalizations to the
case where jobs can be scheduled in a union of time intervals
(in contrast to the usual single release time and deadline).
Under this generalization,
once the capacity constraints exceeds two, minimizing
active time becomes NP-hard via a reduction from 3-EXACT-COVER.

Mertzios et al.~\cite{mertzios2012} consider a dual 
problem to busy time minimization: the \emph{resource 
allocation maximization version}.
Here, the goal is to maximize the number of jobs 
scheduled without violating a budget constraint given
in terms of busy time and the parallelism constraint. 
They show that the maximization version is 
NP-hard whenever the (busy time) minimization problem is NP-hard. 
They give a $6$-approximation algorithm for clique instances and a 
polynomial time algorithm for proper clique instances for the maximization 
problem. 

The online version of both the busy time minimization and resource 
allocation maximization was considered by Shalom et 
al.~\cite{shalom2012online}. They prove a lower bound of $g$ where $g$ 
is the parallelism parameter, for any deterministic algorithm for 
general instances and give an $O(g)$-competitive algorithm. Then they 
consider special cases, and show a lower bound of $2$ and an upper 
bound of $(1+\phi)$ 
for a one-sided clique instances (a special case of laminar cliques),
where $\phi$ is the golden ratio. They 
also show that the bounds increase by a factor of $2$ for clique 
instances. For the maximization version of the problem with parallelism 
$g$ and busy time budget $T$, they show that any deterministic algorithm 
cannot be more than $gT$ competitive. They give a $4.5$-competitive 
algorithm for one-sided clique instances. 

Flammini et al.~\cite{flammini2011optimizing} consider the problem of 
optimizing the cost of regenerators that need to 
be placed on light paths in optical networks, after 
every $d$ nodes, to regenerate the signal. They show that the 
$4$-approximation algorithm for minimizing busy time~\cite{IPDPS09} solves 
this problem for a path topology and $d=1$ and extend it to ring and tree 
topologies for general $d$. 

Faigle et al.~\cite{faigle1996randomized} consider the online problem of 
maximizing ``busy time'' but their objective function is totally different 
from ours. Their setting consists of a single machine and no parallelism. 
Their objective is to maximize the total length of intervals 
scheduled as they arrive online, such that at a given time, at most one 
interval job has been scheduled on the machine. They give a randomized 
online algorithm for this problem.


\section{Active time scheduling of preemptive jobs}

\begin{definition}
A job $j$ is said to be \emph{live} at 
slot $t$ if $t\in [r_j+1 , d_j]$.
\end{definition}

\begin{definition}
A slot is \emph{active} if at least one job is 
scheduled in it. It is \emph{inactive} otherwise. 
\end{definition}

\begin{definition}
An active slot is \emph{full} if there are $g$ jobs 
assigned to it. It is \emph{non-full} otherwise. 
\end{definition}

A feasible solution $\sigma$ is specified by a set 
of active time slots $\textscr{A} \subseteq \textscr{T}$, 
and a mapping or assignment of jobs to time slots in 
$\textscr{A}$, such that at most 
$g$ jobs are scheduled in any slot in $\textscr{A}$, 
at most one unit of any job $j$ is scheduled in 
any time slot in $\textscr{A}$ and every job $j$ 
has been assigned to $p_j$ active slots within its window. 
Once the set $\textscr{A}$ of active 
slots has been determined, a feasible integral assignment can be 
found by computing a max-flow computation on the following
graph.

Define $G_{\mathbf{feas}}$ to be the
flow network whose vertex set is a source $s$, a
sink $v$ and a bipartite subgraph $(\scr{J},\scr{T})$
where $\scr{J}$ contains a node $x_j$ for every job $J_j$
and $\scr{T}$ contains a node for every timeslot, from 
1 to $T$.  For each node $x_j \in \scr{J}$,
there is an edge from $s$ to $x_j$ with capacity
$p_j$.  For any job $j$ that is feasible in slot $t$, 
there is an edge of capacity one between $x_j$ and
$t$'s node in $\scr{T}$.  Finally, there exists an
edge of capacity $g$ between each {\em active} slot node of $\scr{T}$
and the sink (see Figure~\ref{fig:feasflow}). Nodes 
$t_i$ corresponding to inactive slots $i$
may be deleted; alternatively, the capacity of
edge $(t_i,v)$ can be set to $0$.
An active time instance has a feasible schedule 
if and only if the maximum flow that can be sent on 
the corresponding graph $G_{\mathbf{feas}}$
has value $P = \sum_{j = 1}^n p_j$.  (Note that since
capacities are integral, the maximum flow is
also integral without loss of generality.)


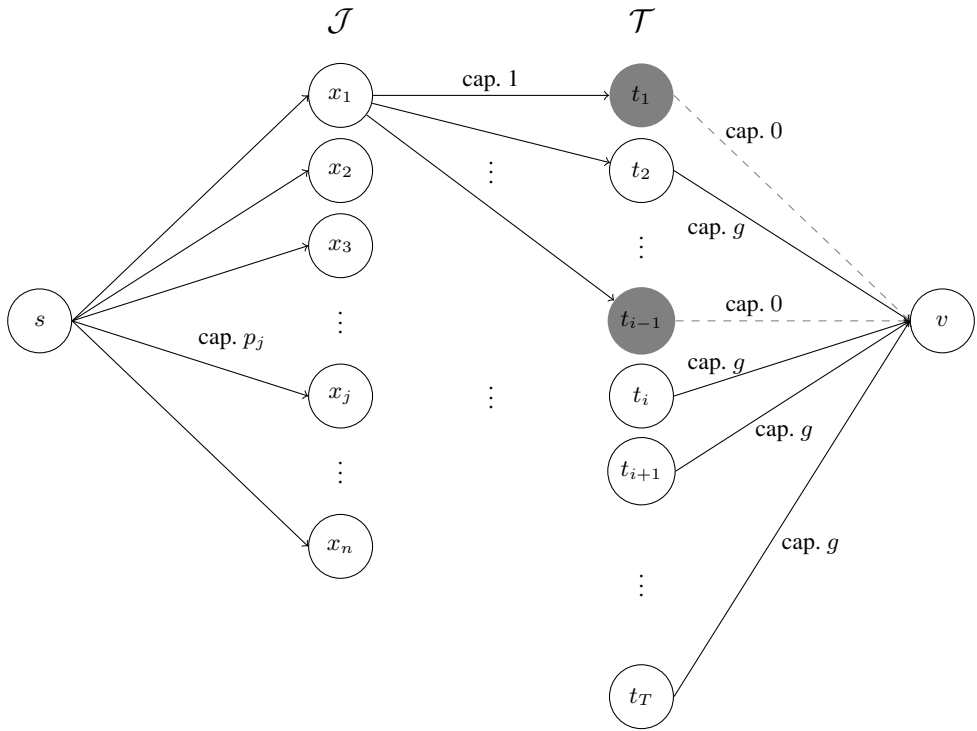
\begin{figure}
\centering
\begin{tikzpicture}[auto]
\tikzstyle{circ}=[draw, circle, minimum size=.85cm]
\tikzstyle{empty_circ}=[fill=gray, circle, minimum size=.85cm]

\node[circ] (s) at (-1,4) {$s$};
\node[circ] (t) at (11,4) {$v$};

\node (J) at (3,8) {\large $\mathcal{J}$};
\node[circ] (j1) at (3,7) {$x_1$};
\node[circ] (j2) at (3,6) {$x_2$};
\node[circ] (j3) at (3,5) {$x_3$};
\node[rotate=90] at (3,4) {\ldots};
\node[circ] (j) at (3,3) {$x_j$};
\node[rotate=90] at (3,2) {\ldots};
\node[circ] (n) at (3,1) {$x_n$};

\node (T) at (7,8) {\large $\mathcal{T}$};
\node[empty_circ] (t1) at (7,7) {$t_1$};
\node[circ] (t2) at (7,6) {$t_2$};
\node[rotate=90] at (7,5) {\ldots};
\node[empty_circ] (1i) at (7,4) {$t_{i-1}$};
\node[circ] (i) at (7,3) {$t_{i}$};
\node[circ] (i1) at (7,2) {$t_{i+1}$};
\node[rotate=90] at (7,0.5) {\ldots};
\node[circ] (tT) at (7,-1) {$t_T$};

\draw [->] (s.east) -- (j1.west);
\draw [->] (s.east) -- (j2.west);
\draw [->] (s.east) -- (j3.west);
\draw [->] (s.east) to node {cap. $p_j$} (j.west);
\draw [->] (s.east) -- (n.west);

\draw [->] (j1) to (t1);
\node at (5,7.2) {cap. 1};
\draw [->] (j1) to (t2);
\node[rotate=90] at (5,6) {\ldots};
\draw [->] (j1) to (1i);

\node[rotate=90] at (5,3) {\ldots};

\draw [->, gray, dashed] (t1.east) -- (t.west);
\node at (8.5,6.5) {cap. $0$};

\draw [->] (t2.east) to (t.west);
\node at (8,5.2) {cap. $g$};

\draw [->, gray, dashed] (1i.east) to (t.west);
\node at (8.5,4.2) {cap. $0$};

\draw [->] (i.east) to (t.west);
\node at (8,3.4) {cap. $g$};

\draw [->] (i1.east) to (t.west);
\node at (8.9,2.5) {cap. $g$};

\draw [->] (tT.east) to (t.west);
\node at (9.25,1) {cap. $g$};

\end{tikzpicture}
\caption{Flow network $G_{\mathbf{feas}}$. 
Nodes in $\mathcal{T}$ corresponding to inactive
slots are gray.  An integral flow of value 
$\sum_j p_j$ corresponds to a feasible schedule.}
\label{fig:feasflow}
\end{figure}

The cost of a feasible solution $\sigma$ is the 
number of active slots in the solution, denoted by $|\textscr{A}|$.  
Let $\textscr{A}_f$ denote the set of active slots that are full, and let 
$\textscr{A}_n$ denote the set of active slots which are non-full. 
Therefore, $|\textscr{A}| = |\textscr{A}_f| + |\textscr{A}_n|$. 

\begin{definition}
A minimal feasible solution is one for which closing any active slot 
renders the remaining active slots an infeasible solution.
In other words, no proper subset of active slots can feasibly
satisfy the entire job set. 
\end{definition}

Given a feasible solution, one can easily find a minimal feasible solution by closing
slots (in any order) and checking if a feasible solution still exists.

\begin{definition}
A non-full-rigid job is one which is scheduled for 
one unit in every non-full slot where it is live. 
\end{definition}

\begin{lemma}
For any minimal feasible solution $\sigma$, there 
exists another solution $\sigma'$ of same cost, 
where every active slot that is non-full, has at least one non-full-rigid 
job scheduled in it. 
\end{lemma}

\begin{proof}
Consider any non-full slot in a minimal feasible solution $\sigma$, which does not have any non-full-rigid job 
scheduled in it. Move any job in that slot to any other (non-full, active) slot that it may be 
scheduled in, and where it is not already scheduled\footnote{Such an action might make the other
slot a full slot, and change the status of a job which could become non-full-rigid.}. There must at least one such slot, otherwise
this would be a non-full-rigid job. Continue this process for as long as possible. Note that in moving 
these jobs, we are not increasing the cost of the solution, as we are only moving jobs to already active slots. 
If we can do this until there are no jobs scheduled in this slot, then we would have found a smaller cost solution, 
violating our assumption of minimal feasibility. Otherwise, there must be at least one job left in that slot, 
which cannot be moved to any other active slots. This can only happen if all the slots in the window of this job 
are either full, or inactive, or non-full where one unit of this job has been scheduled, thus making 
this a non-full rigid job. 

Continue this process until for each non-full slot, 
there is at least one non-full-rigid job scheduled. 
\qed
\end{proof}

\begin{lemma}
\label{lemma:partition}
There exists a minimal set $\textscr{J^*}$ of 
non-full-rigid jobs such that
\begin{enumerate} 
\item at least one of these jobs is scheduled in every non-full slot, and
\item no jobs $J_j$ and $J_{j'}$ exist in $\textscr{J^*}$ such that 
	$J_j$'s window is contained in $J_{j'}$s window, and
\item at every time slot, at most two of the jobs in 
	$\textscr{J^*}$ are live. 
\end{enumerate}

\end{lemma}

Let $OPT$ be the cost of the optimal solution.
This lemma allows us to charge the cost of the non-full slots 
to the quantity $|\textscr{J^*}| \le 2OPT$.
In addition, any optimal solution must have cost at least
the sum of job lengths divided by capacity $g$.
The cost incurred by the full slots is trivially no more than
this latter quantity.  Thus,
once we bound the cost of the non-full slots
by $2OPT$, the final
approximation ratio of three follows (see Theorem~\ref{main:thm1}).

\begin{proof}
Consider a set $\textscr{J^*}$ of non-full-rigid 
jobs that are covering all the non-full slots. 
Suppose it contains a pair of non-full-rigid jobs 
$j$ and $j'$, such that the $[r_j, d_j) \subseteq [r_{j'}, d_{j'})$. 
One unit of $j'$ must be scheduled in every non-full 
slot in the window of $j'$. However, this 
also includes the non-full slots in the window of $j$, 
hence we can discard $j$ from $\textscr{J^*}$ without loss.  
We repeat this with every pair of non-full-rigid jobs in $\textscr{J^*}$, such that the window of one 
is contained within the window of another, until there exists no such pair. 

Now consider the first time slot $t$ where $3$ or more jobs of $\textscr{J^*}$ are live. Let these jobs be 
numbered according to their deadlines ($j_1, j_2, j_3, \ldots. j_{\ell}, \ \ell\geq 3$). 
By definition, the deadline of all of these jobs 
must be at least $t$ since they are all 
live at $t$. Moreover, they are all non-full-rigid, 
by their membership in $\textscr{J^*}$, 
which means they are scheduled in every non-full active slot 
in their window.  Since $\textscr{J^*}$ is minimal, 
no job window is contained within another, 
hence none of the jobs $j_2, \ldots, j_{\ell}$ have 
release time earlier than that of $j_1$. Therefore, all non-full slots before the deadline 
of $j_1$ must be charging either $j_1$ or some other job with an earlier release time. Consequently,  
discarding any of the jobs $j_2, \ldots, j_{\ell}$ will not affect the charging of these slots. 
  
Let $t'$ be the first non-full active slot after the deadline of $j_1$; 
then $t'$ must charge one of $j_2, j_3, \ldots, j_{\ell}$.  
Among these, all jobs which have a deadline earlier than $t'$, 
can be discarded from 
$\textscr{J^*}$, without any loss, since no non-full 
slot needs to charge it. 
Hence, let us assume that all of these jobs 
$j_2, j_3,\ldots, j_{\ell}$ are live at $t'$. 
However, all of them being non-full-rigid, and $t'$ 
being non-full and active, all of them must have 
one unit scheduled in $t'$. Therefore, if we discard 
all of the jobs $j_2, \ldots, j_{\ell-1}$ and 
keep $j_{\ell}$ alone, that would be enough since it can be charged 
all the non-full slots between $t'$ and its deadline $d_{\ell}$. 
Hence, after discarding these intermediate jobs from 
$\textscr{J^*}$, there would be 
only two jobs $j_1$ and $j_{\ell}$ left which overlap at $t$. 

Repeat this for the next slot $t''$ where $3$ or more 
jobs of $\textscr{J^*}$ are live, until there 
are no such time slots left. 
\qed
\end{proof}

The cost of the non-full slots of the minimal feasible solution $\sigma'$ is $|\textscr{A}_n| \leq \sum_{j \in \textscr{J^*}}{p_j}$. 

\begin{theorem}
\label{main:thm1}
The cost of any minimal feasible solution is at most $3OPT$.
\end{theorem}
\begin{proof}
It follows from Lemma \ref{lemma:partition} that $\textscr{J^*}$ can be partitioned into two job sets $\textscr{J}_1$ and $\textscr{J}_2$, 
such that the jobs in each set have windows disjoint from one another. 
Therefore the sum of the processing times of 
the jobs in each such partition is 
a lower bound on the cost of any optimal solution. 
Hence, the cost of of the non-full slots is 
$|\textscr{A}_n| \leq \sum_{j \in \textscr{J^*}}{p_j} \leq \sum_{j\in \textscr{J}_1}{p_j} + \sum_{j'\in \textscr{J}_2}{p_{j'}} \leq 2OPT$. 
Furthermore, the full slots charge once to $OPT$, since they have a mass of $g$ scheduled in them, which is a lower bound on 
$OPT$. $|\textscr{A}_f| \leq \frac{\sum_{j \in \textscr{J}}{p_j}}{g} \leq OPT$. 
Therefore, in total the cost of any minimal feasible solution $cost(\sigma) = cost(\sigma') = |\textscr{A}|= 
|\textscr{A}_f| +|\textscr{A}_n| \leq 3OPT$. 
This proves the theorem.
\qed
\end{proof}

The above bound is asymptotically tight as demonstrated
by the following example. 
There are two jobs each of length $g$.  One has
window $[0,2g)$ and the other has window $[g,3g)$. 
Also, there are $g-2$ rigid jobs, each of length $g-2$, with windows
$[g+1, 2g-1)$. Finally, there are $g-2$ unit-length jobs with windows 
$[g+1, 2g)$ and another $g-2$ unit-length jobs with windows $[g, 2g-1)$. 
See Figure~\ref{fig:activetimeexample}.
An optimal solution schedules the two longest jobs in $[g, 2g)$, 
one set of $g-2$ unit-length jobs at time slot $g$, 
and the other set of unit-length jobs at time slot $2g-1$,
for an active time of $g$.
However, a minimal feasible solution could schedule the 
two sets of $g-2$ unit-length jobs in the window $[g+1,2g-1)$, 
with the rigid jobs of length $g-2$. Now, the two longest 
jobs cannot fit anywhere in the window $[g+1, 2g-1)$, 
since these slots are full.
The minimal feasible solution must put these jobs somewhere; 
one feasible way 
would be to pack one of the longest 
jobs from $[1,g+1)$ and the other one from $[2g-1,3g-1)$. 
The total cost of this minimal feasible solution
would then be $3g-2$, which approaches
to $3OPT$ as $g$ increases.

\begin{figure}[htbp]
	\centering
	\includegraphics[width=0.75\textwidth]{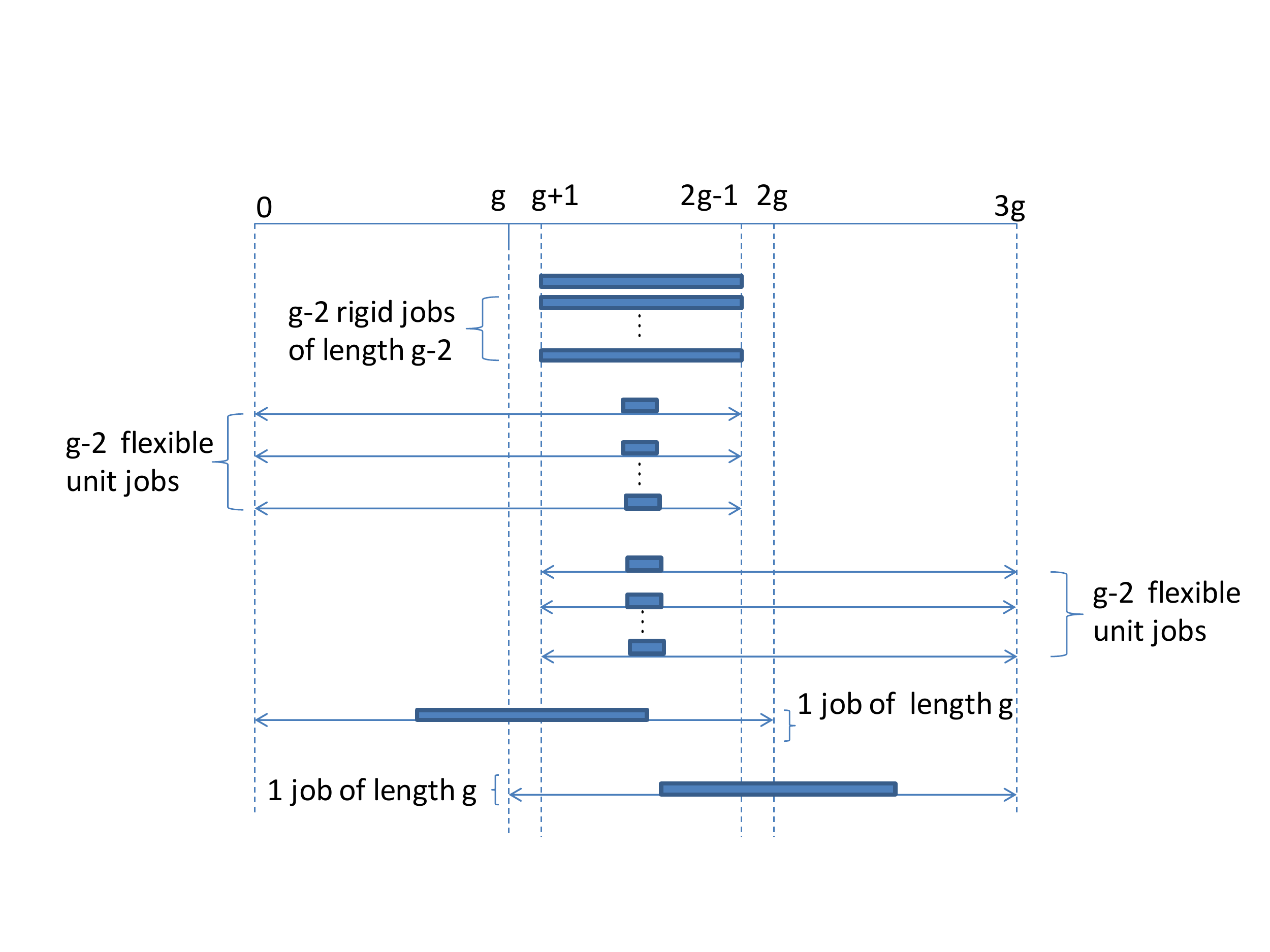}
	\caption{An instance where a minimal feasible solution 
		has active time almost 3 times the optimal solution.}
	\label{fig:activetimeexample}
\end{figure}

%
%
\section{A $2$-approximation LP rounding algorithm} 
In this section, we give a 2-approximate LP-rounding algorithm 
for the active time problem with non-unit length jobs,
where preemption is allowed at integral boundaries. 
Consider the following integer program for this problem. 

\begin{equation*}
\begin{array}{lrll}
IP: \text{min.} & \displaystyle\sum\limits_{t \in \mathcal{T}} y_t 
    & & \text{s.t.}\\[2.0ex]
& x_{t,j} &\le y_t \hspace{1cm} &\forall \ t \in \textscr{T}, j \in 
    \textscr{J}\\[2.0ex]
& \displaystyle\sum\limits_{j \in \textscr{J}} x_{t,j} 
    &\le g \cdot y_t & \forall \ t \in \textscr{T} \\[2.0ex]
& \displaystyle\sum\limits_{t \in \textscr{T}} x_{t,j} 
    &\ge p_j & \forall \ j \in \textscr{J}\\[2.0ex]
& y_t &\in \{0,1\} & \forall \ t \in \textscr{T}\\[2.0ex]
& x_{t,j} &\in \{0,1\} & \forall \ t \in \textscr{T}, j \in \textscr{J}\\[2.0ex]& x_{t,j} &= 0 & \forall \ t \notin \{r_j+1,\ldots,d_j\}
\end{array}
\end{equation*}

In the integer program, the indicator variables $y_t$
denote whether slot $t$ is active (open).  
The assignment variables $x_{t,j}$ 
specify whether any unit of job $j$ is assigned to slot $t$.
The first set of inequalities ensures that a unit of any job
can be scheduled in a time slot only if that slot is active.
Without this constraint, the LP relaxation of $IP$ has an
unbounded integrality gap. This constraint,
along with the range specification on the 
indicator $y$ variables, also ensures
that at most one unit of a job can be assigned to a single slot.
The second set of inequalities ensures that at most
$g$ units of jobs can be assigned to an active slot.
The third set of inequalities ensures that $p_j$ units
of a job $j$ get assigned to active slots.
The remaining constraints are range specifiers for the indicator and
assignment variables.
We relax the integer requirements of $y_t$ and $x_{t,j}$
to get the LP relaxation $LP1$. In $LP1$, the objective function
and all the constraint inequalities remain unchanged from the integer program $IP$,
except the range specifiers at the end.  Those get modified
as follows: $0\leq y_t\leq 1$ for every slot $t$,
and $x_{t,j} \geq 0$ for every job $j$ and slot $t$.
As in the IP, $x_{t,j}$ is forced to 0 for slots $t$ not 
in job $j$'s window $\{r_j+1, \ldots, d_j\}$.
Henceforth, we focus on $LP1$.


We first solve $LP1$ to optimality. Since any
integral optimal solution is a feasible LP solution,
the optimal LP solution is a lower bound on the cost of
any optimal solution. Our goal is to round it to a feasible
integral solution that is within twice the cost of
the optimal LP solution.  However, before we do the rounding,
we preprocess the optimal LP solution so that it has
a certain structure without increasing the cost of the solution.
Then we round this solution to obtain an integral feasible solution.
We use the following notation to denote the LP solution.
A slot $t$ with $y_t = 1$ is said to be \textit{fully open}, 
a slot with $1 > y_t \geq \frac{1}{2}$
is said to be \textit{half open}, 
a slot with $\frac{1}{2}> y_t >0$ is said to be \textit{barely open} and
a slot with $y_t = 0$ is said to be \textit{closed}.

In the rounding, our goal will be to find a set of slots
to open integrally, such that there exists a feasible fractional assignment for the jobs in the
integrally open slots. As described earlier, given a set of integrally open slots with
feasible fractional assignment of jobs for the active time problem, 
an integral assignment can be found at the end of the rounding 
procedure, via a maximum flow computation. 
Since the capacities are integral and integrality of
flow, the flow computation returns an integral assignment, without
loss of generality. 

\subsection{Preprocessing: Creating a Right-Shifted Solution}
First, consider the set of distinct deadlines
$\textscr{D} = \{t_{d_1}, t_{d_2}, \ldots, t_{d_{\ell}}\}$,
sorted in increasing order. 
We will process the LP solution sequentially according
to this order.
Denote by $\textscr{J}_i$ the set of jobs with deadline $t_{d_i}$,
and define a dummy deadline $t_{d_0}\leq t_{d_1}$ to be
the earliest slot $t$ where $y_t >0 $ in the optimal solution $LP1$.
Note that $t_{d_0}$ does not correspond to an actual deadline
(and hence, $\mathcal{J}_0 = \emptyset$), but is defined
simply for ease of notation.
If $t_{d_0} < t_{d_1}$, we add $t_{d_0}$ to $\mathcal{D}$,

Next, we consider the slots open between successive deadlines in the
optimal solution for $LP1$.
$Y_i$ is the sum of $y_t$ over the slots numbered
$\{t_{d_{i-1}}+1, \ldots, t_{d_i}\}$ for all $i\geq 1$.
For the rest of the paper, $[q]$ is
shorthand notation for $\{1, 2, \ldots, q\}$.

\begin{definition}
$Y_i = \sum_{t\in \{t_{d_{i-1}}+1, \ldots, t_{d_i}\}}{y_t}, \ \forall\ i \in [\ell]$,
where $\ell$ is the number of distinct deadlines in $\mathcal{J}$, and $Y_0 = 0$.
\end{definition}

By definition, the cost of the LP solution is:
$\sum_{t\in \textscr{T}}{y_t} = \sum_{t_{d_i} \in \mathcal{D}}{Y_i}$.
Now, we modify the optimal LP solution as follows to
create a right-shifted structure.  Intuitively, we want
to push open slots to the ``right'' (i.e., delay them to later time
slots) as much as possible without violating feasibility, or changing $Y_i$.

%
For all $i\in [\ell]$,
we open the slots $\{t_{d_i}-\lfloor Y_i \rfloor+1, \ldots, t_{d_i}\}$
integrally, and the slot $t_{d_i}-\lfloor Y_i \rfloor$
partially, up to
$Y_i- \lfloor Y_i \rfloor$, if $Y_i - \lfloor Y_i \rfloor > 0$,
and close all slots 
$\{t_{d_{i-1}}+1, \ldots, t_{d_i}-\lfloor Y_i \rfloor - 1\}$.

Using the notation stated earlier, slots
$\{t_{d_i}-\lfloor Y_i \rfloor+1, \ldots, t_{d_i}\}$ are fully open, 
and the slot $t_{d_i}-\lfloor Y_i \rfloor$ is either closed 
if $Y_i- \lfloor Y_i \rfloor = 0$,
barely open if $\frac{1}{2}> Y_i- \lfloor Y_i \rfloor > 0$, 
or half open if $1> Y_i- \lfloor Y_i \rfloor \geq \frac{1}{2}$.
Refer to the Figure \ref{fig:rightshifted} for an example.

\begin{figure}
    \centering
        \includegraphics[width=0.75\textwidth]{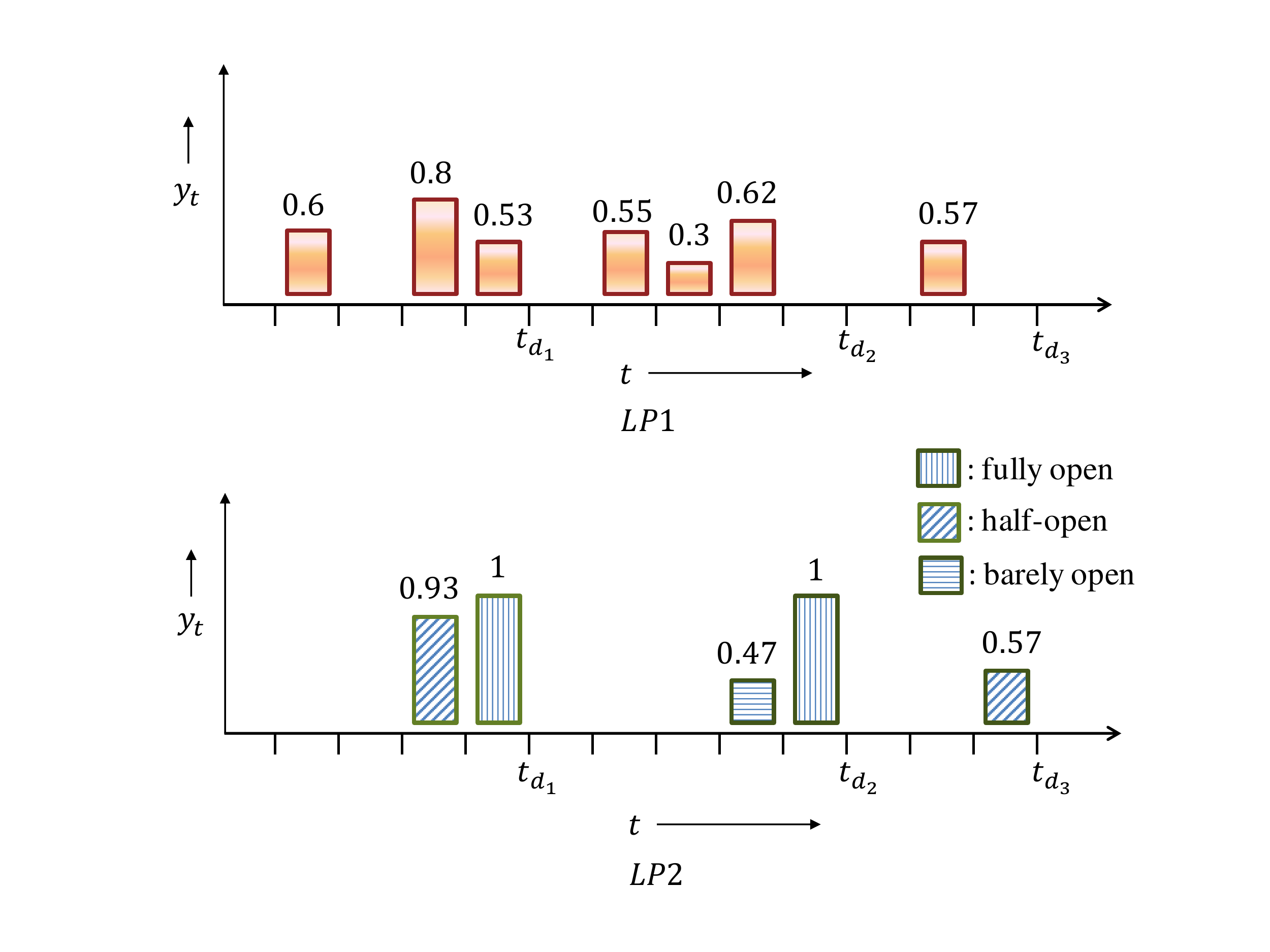}
        \caption{Figure showing the right shifted solution $LP2$ corresponding to $LP1$.}
    \label{fig:rightshifted}
\end{figure}

We now prove that there exists a feasible fractional 
assignment of all the jobs in the right-shifted LP solution.
\begin{lemma}
All jobs in $\mathcal{J}$ can be feasibly fractionally assigned in the right-shifted LP solution.
\end{lemma}
\begin{proof}
For any $t_{d_i}> t_{d_0}$, 
$Y_i$ is unchanged in the right-shifted LP solution.
Consider the following processing of the original LP solution for
a pair of slots $t_{d_i}$ and $t_{d_i}-1$ for some $i>0$.
If $y_{t_{d_i}}<1$ in the original LP solution,
let $\delta = y_{t_{d_i}} + y_{t_{d_i}-1} -1$. 
If $\delta\leq 0$, move the job assignments as is
from $y_{t_{d_i}-1}$ to $y_{t_{d_i}}$, 
and update the $x$ variables of the respective jobs
moved to reflect the new slot they are assigned to. In other words,
$x_{t_{d_i},j}$ is incremented by $x_{t_{d_i}-1,j}$ and $x_{t_{d_i}-1, j}$
is reduced to $0$. At the same time, also increment the $y_{t_{d_i}}$
by $y_{t_{d_i}-1}$ and decrement $y_{t_{d_i}-1}$ to $0$. 
This does not violate LP feasibility.

If $\delta >0$, increment $y_{t_{d_i}}$ by $1 - y_{t_{d_i}}$, 
and decrement $y_{t_{d_i}-1}$ by
$1 - y_{t_{d_i}}$. The new values are $y'_{t_{d_i}} = 1$ 
and $y'_{t_{d_i}-1} = \delta$.
For every job $j$ with a positive assignment to $t_{d_i}-1$,
decrement $x_{t_{d_i}-1,j}$ by
$\delta' = \max (0, x_{t_{d_i}-1,j}-\delta)$, and increment $x_{t_{d_i},j}$ by $\delta'$.
By this transformation, for every job $j$, the updated $x_{t_{d_i}-1, j} \leq y'_{t_{d_i}-1}$,
and $x_{t_{d_i},j}\leq y'_{t_{d_i}}$. Moreover, the total mass of jobs transferred to $t_{d_i}$
is at most $\sum_{j\in \mathcal{J}}\left({x_{t_{d_i}-1,j}-\delta}\right) \leq g (1 - y_{t_{d_1}})$.
Hence, the total mass of jobs in $t_{d_i}$ is at most $g$ and at $t_{d_i}-1$ is at most
$g\delta$.
Again, this maintains LP feasibility.

Now, repeat this process for the updated $y_{t_{d_i}-1}$ 
and $y_{t_{d_i}-2}$.
Continue this till the pair of adjacent slots $t_{d_{i-1}+1}$ 
and $t_{d_{i-1}+2}$ are
processed. We end up with a feasible right-shifted LP solution
for time slots $\{t_{d_{i-1}+1}, \ldots, t_{d_i}\}$. 
Repeat this for all $i \in [\ell]$.
We get a feasible fractional right-shifted LP solution.
\qed
\end{proof}

We can solve a feasibility LP with the $y_t$ for any $t$ pre-set, as dictated by the right-shifting process. The right-shifted LP
is described below.

%
\begin{equation}
\begin{array}{rrll}
LP2: & x_{t,j} &\leq y_t & \forall \ t \in \textscr{T}, j \in \textscr{J}\\[5pt]
 & \sum_{j\in \textscr{J}}{x_{t,j}} &\leq g \cdot y_t \ \ \ & \forall \ t \in \textscr{T}\\[5pt]
& \sum_{t\in \textscr{T}}{x_{t,j}} &\geq p_j & \forall \ j \in \textscr{J}\\[5pt]
& x_{t,j} &\geq 0 & \forall \ t \in \textscr{T}, j \in \textscr{J} \\[5pt]
& x_{t,j} &= 0 & \forall \ t \notin \{r_j+1, \ldots, d_j\}
\end{array}
\end{equation}

Henceforth, we work with this feasible, right-shifted optimal LP solution, $LP2$.

\begin{observation}
In a right-shifted fractional solution, a slot
$t$, $t_{d_{i-1}} < t < t_{d_{i}}$ for any $i \in [\ell]$,
is fully open only if slots
$\{t+1, \ldots, t_{d_{i}}\}$ are fully open.
\label{obs:rightshifted}
\end{observation}

The above observation follows from the description of the
preprocessing step to convert an optimal LP solution
to a right-shifted solution structure.

\subsection{Overview of Rounding}
In this section, we will give an informal overview of the rounding
process for ease of exposition. The formal description of rounding
and the proofs are given in the following sections.

We process the set of $\ell$ distinct deadlines $\textscr{D}$,
in increasing order of time. In each iteration $i \ \in [\ell]$, we
consider the jobs in $\mathcal{J}_i$, and we will
integrally open a subset of slots, denoted
$\textscr{O}_i \subseteq \{t_{d_{i-1}}+1, \ldots, t_{d_i}\}$,
maintaining the following invariants at the end of every iteration $i$:
(i) there exists a feasible assignment of all jobs in $\bigcup_{k \in [i]}{\mathcal{J}_k}$
in the set of integrally open slots thus far $\bigcup_{k \in [i]}{\textscr{O}_k}$;
(ii) the total number of integrally open slots thus far
is at most twice the cost of LP solution thus far, i.e.,
$|\bigcup_{k \in [i]}{\textscr{O}_k}|\leq 2 \sum_{k \in [i]}{Y_k}$.

Let the cost of a slot $t$ after rounding be denoted as $y'_t$.
Note that while \emph{fully open} refers to any
slot $t$, such that $y_t=1$, \emph{integrally open} refers to any
slot $t'$ such that after rounding, $y'_{t'} = 1$.

\begin{definition}
For all $i \in [\ell]$,
the set of integrally open slots $\mathcal{O}_i = \{t \in \{t_{d_{i-1}}+1, \ldots, t_{d_i}\} \ | \ y'_t = 1\}$, where
$y'_t$ is the rounded value of $y_t$ for any slot $t$.
\end{definition}

In our rounding scheme, every fully open slot in the
right-shifted solution will also
be integrally open; however, there may be integrally open
slots that were not fully open in the fractional optimal solution.
These slots will need to be accounted for carefully.

%
The rounding algorithm is as follows.
At every iteration $i$, it first opens $\lfloor Y_i\rfloor$ slots integrally,
going backwards from $t_{d_i}$ in the right-shifted solution.
Note that these slots were fully open in the right-shifted solution,
and hence, obviously, do not charge anything extra
to the LP solution. If $Y_i - \lfloor Y_i\rfloor = 0$,
we are done with this iteration. Otherwise, if $Y_i - \lfloor Y_i \rfloor \geq \frac{1}{2}$,
there is one half-open slot $t_{d_i} - \lfloor Y_i\rfloor$,
which the rounding opens up integrally. This is fine since a half-open slot can be
integrally opened, incurring a rounded cost of $1$, and thereby,
charging their fractional LP cost at most $2$ times.
In case, $Y_i - \lfloor Y_i \rfloor < \frac{1}{2}$, the slot $t_{d_i} - \lfloor Y_i\rfloor$
is barely open.
The rounding algorithm first tries to close a barely open
slot when it is processing such a slot,
by trying to accommodate all the jobs with deadlines
up to the current one, in the current set
of integrally open slots. If no such assignment exists,
then the algorithm opens up such a slot
integrally, after accounting for its value,
by charging other fully open or half open slots.
We next describe the possible ways of charging a barely open slot that needs to be
opened by the rounding.

When a single barely open slot charges a fully open slot,
we say that the barely open slot is \textit{dependent}
on the fully open slot. Now, the barely open slot can
be opened integrally charging the fully open slot twice,
and not charging the barely open slot at all.
First, the algorithm tries to charge a barely open slot
(that needs to be opened) to the earliest fully open
slot without a dependent. Note here fully open slot refers
to slots fully open in the right-shifted LP solution,
and not the set of slots opened integrally after rounding.

If all the earlier fully open slots have dependents, the algorithm
tries to find the earliest fully open slot with a dependent, such that
the sum of the $y$'s of the barely open slot
(that is currently being processed) and the
dependent of the fully open slot add up to at least $\frac{1}{2}$.
Once the algorithm finds the earliest fully open slot with a dependent
that satisfies the above condition, it opens up the barely open slot,
since the cumulative sum of the $y$'s of the three slots is at least
$\frac{3}{2}$. We refer to such a triple of slots as a \emph{trio}. Note that
this trio or triple of slots cumulatively charges at most twice to the cumulative
$y$ (LP cost) of the trio.

Finally, if no trio is possible or all of the earlier fully open slots
are already part of trios, we find the earliest half-open slot,
such that the sum of the $y$ of the half open and the barely open slot
together is at least $1$. Note that here we require
that the half-open slot be an earlier slot that is
already processed. In fact, as a rule of thumb,
the rounding algorithm can only charge slots already processed.
We say that the barely open slot is a \textit{filler} of the
half open slot in this case, and open up the barely open slot.
Note that the two integrally open slots corresponding
to the half open slot and its filler
charge their LP cost at most twice in the process.

In the first case, where a barely open slot
is a dependent on a fully open slot, we do not charge its $y$ value
or LP cost at all; however, in both the latter cases (trio and filler),
we must charge the $y$ or LP cost of all the
slots involved including the barely open one.
A barely open slot that is charged in one iteration as a
dependent (when its $y$ value
was not charged at all) on a fully open slot
can get charged as a trio (when its $y$ value
gets charged) in a later iteration. Similarly, a half-open slot,
charging itself alone in an iteration, can later
get charged with a filler.

Refer to Figure \ref{fig:lprounding} for an example of the
charging process of a barely open slot that cannot 
be closed by flow-based assignments.

\begin{figure}
    \centering
        \includegraphics[width=0.75\textwidth]{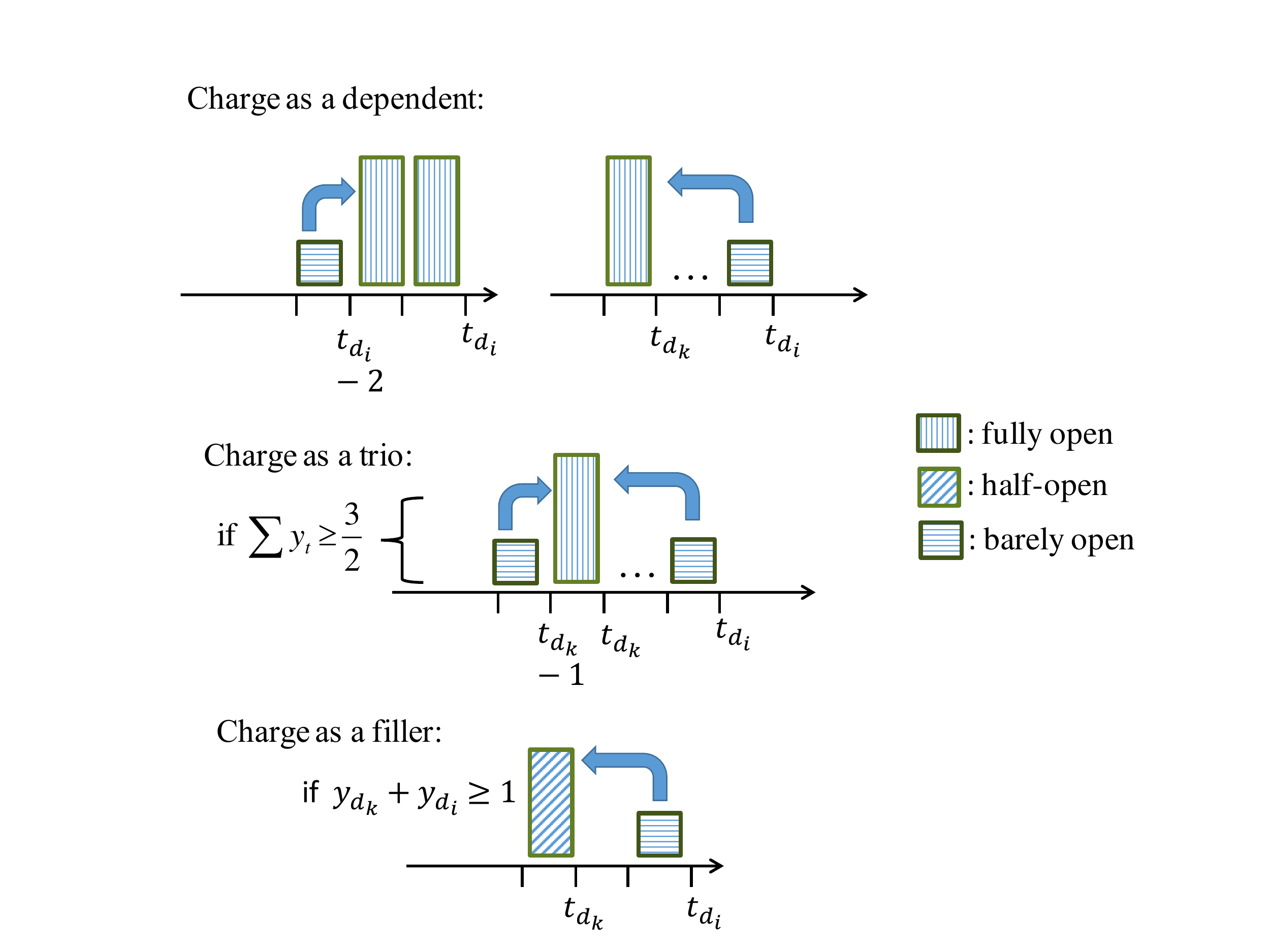}
    \caption{Figure showing the three ways of charging of a barely open slot that needs to be opened by the rounding algorithm.}
    \label{fig:lprounding}
\end{figure}

We will additionally maintain the invariant that at
every iteration, every barely open
slot that we have opened is either a dependent on a fully open slot, or
is part of a trio, or is a filler of a half-open slot. Moreover,
we ensure that every fully open slot has at most one dependent or it is
part of at most one trio, and every half-open slot has
at most one filler.

As already mentioned, in the rounding process we
sometimes close a barely open slot $t_{d_i} - k$, 
while processing a deadline $t_{d_i}$.
This will be done only if the jobs in
$\bigcup_{k \in [i]}{\textscr{J}_k}$
can be fully accommodated in the slots
$\bigcup_{k \in [i]}{\textscr{O}_k}$.
At the same time, the cost of $y_{t_{d_i}-k}$ in the
LP solution is not charged at all.  However, the LP solution
might have assigned jobs of a later deadline in this
barely open slot and hence, we might need to open up
this slot later to accommodate such jobs.
In order to do this, and account for the
the total charge on the LP solution, we create a
``proxy'' copy of the slot that we closed, and carry it over
to the next iteration. The idea behind carrying
forward this ``proxy'' slot is that, if needed,
we can come back and open up this slot in the future,
while processing the jobs of a later deadline,
and account for it without double counting.
Informally, this is a safety deposit that we can
come back and use if needed.
The $y$ cost of this proxy slot $p= t_{d_i}-k$ is
denoted as $y''_p = y_{t_{d_i}-k}$, that is,
the $y$ of the slot we have just closed.
The proxy points to the actual slot $p$
that it is a proxy for, so that we know where
we can open up a slot charging the proxy value, if needed.

In any iteration $i$, when we have a proxy (which by definition
is a barely open slot), we treat it as a regular fractionally open slot (though
there may not be any actual slot at that point). If this slot remains closed after the rounding, the proxy
gets carried over to the next iteration, whereas if it does get opened by the rounding,
the actual slot which it points to gets opened.
However, now the cost of opening it will be accounted for by the current solution.
It may also happen that the proxy from the previous iteration $i$ gets merged with a
fractionally open slot in the current iteration $i+1$. This does not affect the
feasibility of job assignments, since the jobs of a later deadline that were assigned by the
LP in some slot $t_{d_i}-k$, are also feasible to be assigned in slots $t_{d_i} < t' \leq t_{d_{i+1}}$.
This is outlined in detail in Section~\ref{sec:proxy}.
There can be at most
one proxy slot at any iteration.

We next give a detailed description of the rounding process.
\subsection{Processing $t_{d_1}$}
\label{sec:d1}
\begin{claim}
$Y_1\geq 1$.
\end{claim}

\begin{proof}
This is obvious as a feasible LP solution would have
assigned at least one unit of the jobs with deadline
$t_{d_1}$ in slots $t\leq t_{d_1}$.
\qed
\end{proof}

Slots $\{t_{d_1} - \lfloor Y_1 \rfloor +1, \ldots, t_{d_1}\}$
are fully opened in the right-shifted optimal LP solution.
We keep these slots open, i.e., $y'_t = y_t = 1$ for these slots.
In the following, we outline how we deal with the slot
$t_{d_1}- \lfloor Y_1 \rfloor$, if $Y_1 - \lfloor Y_1 \rfloor > 0$.

\begin{customcase}{1}
$Y_1 - \lfloor Y_1 \rfloor \geq \frac{1}{2}$.
\end{customcase}

In this case, the slot $t_{d_1}-\lfloor Y_1 \rfloor$
is half open. We open it fully in the rounding process,
in other words, set $y'_{t_{d_1}-\lfloor Y_1 \rfloor} = 1$,
and charge the cost of fully opening it to
$y_{t_{d_1}- \lfloor Y_1 \rfloor}$. Therefore, the LP cost
$y_{t_{d_1}- \lfloor Y_1 \rfloor}$ is charged
at most twice by this process.

\begin{customcase}{2}
$0<Y_1 - \lfloor Y_1 \rfloor < \frac{1}{2}$.
\end{customcase}

In this case, the slot $t_{d_1}-\lfloor Y_1 \rfloor$ is barely open.
The rounding algorithm first checks if a
feasible assignment of $\textscr{J}_1$ exists in the
slots $\{t_{d_1} - \lfloor Y_1 \rfloor +1, \ldots, t_{d_1}\}$
(the fully open slots), using the maximum-flow
construction described earlier.
If such an assignment exists, the slot $t_{d_1}-
\lfloor Y_1 \rfloor$ is closed, that is,
$y'_{t_{d_1}-\lfloor Y_1 \rfloor} = 0$, and
we proceed to the next iteration (processing $t_{d_2}$),
after passing over a proxy slot with $y= Y_1 - \lfloor Y_1\rfloor$
to iteration $2$. Note that $y_{t_{d_1}-\lfloor Y_1 \rfloor}$
is not charged at all so far by the rounding process.
We denote the proxy as $y''_p = y_{t_{d_1}-\lfloor Y_1 \rfloor}$,
where $p = t_{d_1}-\lfloor Y_1 \rfloor$.

Otherwise, if a feasible assignment does not exist,
the rounding algorithm opens it fully, that is,
$y'_{t_{d_1}-\lfloor Y_1 \rfloor} = 1$; its cost needs
to be accounted for. In order to account for the cost
of opening it, we consider the slot to be dependent
on $t_{d_1}-\lfloor Y_1 \rfloor + 1$, the earliest
fully open slot without a dependent.  We charge
the cost of this slot to $y_{t_{d_1}-\lfloor Y_1 \rfloor + 1}$.
By this process,
$y_{t_{d_1}-\lfloor Y_1 \rfloor + 1}$ is charged at most twice.
We are guaranteed to find a fully open slot
on which to make it dependent since $Y_1 > 1$ in this case.

We next argue that there exists a feasible assignment
of all jobs with deadline $t_{d_1}$ in the slots opened
by the above rounding algorithm up to deadline $t_{d_1}$. This will form the base case of the
argument for the general case that we will prove by induction.

\begin{lemma}
\label{lemma:basecase}
There exists a feasible integral assignment of all jobs with deadline $t_{d_1}$
in the slots opened by the rounding algorithm up to deadline $t_{d_1}$.
Moreover, $|\mathcal{O}_1| \leq 2 Y_1$.
\end{lemma}

\begin{proof}
Suppose by contradiction that flow could not find an assignment of all jobs with deadline
$t_{d_1}$ in the slots opened by the rounding algorithm up to $t_{d_1}$.
Consider the following assignment. Assign the jobs with release time $t_{d_1}-1$ 
to slot $t_{d_1}$. If the slot gets full, then move on to $t_{d_1}-1$. Otherwise, next assign
jobs with release time $t_{d_1}-2$, till $t_{d_1}$ gets full, then move on to $t_{d_1}-1$.
Continue this till $t_{d_1} - \lfloor Y_1\rfloor$ slots. If flow
could not find an assignment, clearly, all slots are completely full,
and still there is at least one job left. Therefore, there are
$\geq (\lfloor Y_1\rfloor +1)g+1$ jobs with
deadline $t_{d_1}$. However, the LP cost up to $t_{d_1}$ is $Y_1$, hence the LP could not
have scheduled more than $Y_1 g \leq (\lfloor Y_1\rfloor +1)g$ jobs in $t_{d_1}$.
Therefore, this gives a contradiction.

Now, we know that $Y_1 \geq 1$, otherwise LP would not be feasible.
Moreover, by the right-shifted nature, there can be
only one barely open slot in $Y_1$. Clearly, $|\mathcal{O}_1| \leq 2 Y_1$, and if at all, a barely open
slot was rounded to integrally open, we have at least one fully open slot
to charge it to.
This completes the proof.
\qed
\end{proof}

\subsection{Processing deadline $t_{d_i}$, $i>1$}
\label{sec:proxy}

Next we describe the rounding process for an
arbitrary deadline $t_{d_i}$, $i>1$, after the previous
deadlines $\{t_{d_1}, \ldots, t_{d_{i-1}}\}$ have been processed.
Since we are working with a right shifted solution,
the slots $\{t_{d_i}-\lfloor Y_i \rfloor +1, \ldots, t_{d_i}\}$
are fully open.
Hence, these slots will remain integrally open in the rounded solution,
i.e., for these slots, $y'_t = y_t$ for
$t \in \{t_{d_i}-\lfloor Y_i \rfloor +1, \ldots, t_{d_i}\}$.

\subsubsection{Dealing with a proxy slot.}
While processing a deadline $t_{d_i}$, suppose there is a proxy of value $y''_p$ 
carried over from iteration $(i-1)$. Here, $p$ denotes the slot
pointed to by the proxy.

Before we do any rounding in iteration $i$, we merge the proxy with $Y_i$
in the following way.

\begin{customcase}{1}
$y''_p + Y_i - \lfloor Y_i \rfloor \leq 1$.
\end{customcase}

Create a new proxy $y''_{p'} = y''_p + Y_i - \lfloor Y_i \rfloor$.
If $t_{d_{i-1}} \neq t_{d_i}-\lfloor Y_i\rfloor$
(this will always hold if $Y_i - \lfloor Y_i \rfloor > 0$,
and may or may not hold otherwise), we set
$p' = t_{d_i}-\lfloor Y_i\rfloor$.
Otherwise, we set $p' = p$ (i.e., keep the pointer unchanged).
We remover the earlier proxy, or, set $y''_p = 0$.
After this, we consider $Y_i = Y_i+y''_{p'}$.
This changing of the proxy pointer is without loss of generality as
explained next. Clearly,
the proxy cost in iteration $i-1, i>1$, is incurred by the LP in accommodating jobs of a later
deadline and hence it was passed over to iteration $i$. Further, observe that
the jobs of a later deadline that were feasible in $t_{d_{i-1}}$ or earlier, are also
feasible at $t_{d_i} - \lfloor Y_i \rfloor$.

If $y''_p + Y_i - \lfloor Y_i \rfloor = 1$,
then $t_{d_i} - \lfloor Y_i \rfloor$ is
now fully open (hence gets added to $\mathcal{O}_i$) and
no proxy will get carried over
from this iteration\footnote{Note that this can only happen if
$Y_i - \lfloor Y_i \rfloor > \frac{1}{2}$, in other words,
if the slot $t_{d_i} -\lfloor Y_i\rfloor$ is half open, since by
definition a proxy is barely open.}.

Otherwise, we process $p'$
as a regular fractional slot, that is barely open if $y''_p + Y_i - \lfloor Y_i \rfloor < \frac{1}{2}$,
and alternatively, half open if $1> y''_p + Y_i - \lfloor Y_i \rfloor \geq \frac{1}{2}$.

If the slot $p'$ pointed to by the proxy
gets opened by the rounding process ($y'_{p'}$ is set to $1$),
then no proxy is carried over from the iteration $i$.
Otherwise, the new proxy carried over will be of
value $y''_{p'}$ and it will continue to point to slot $p'$.

\begin{customcase}{2}
$y''_p + Y_i - \lfloor Y_i \rfloor> 1$.
\end{customcase}
By definition of proxy, $y''_p< \frac{1}{2}$.
Therefore, this case implies that
$Y_i - \lfloor Y_i \rfloor > \frac{1}{2}$, and hence,
it holds that there exists a slot
$t_{d_i} - \lfloor Y_i \rfloor \neq t_{d_{i-1}}$, that is half open.
We create a new proxy of cost
$y''_{p'} = y''_p + Y_i - \lfloor Y_i \rfloor -1$,
setting the earlier proxy cost to $y''_p = 0$.
If a slot $t_{d_i} - \lfloor Y_i \rfloor - 1 \neq t_{d_{i-1}}$
exists, we point the new proxy to this slot,
in other words, $p'' = t_{d_i} - \lfloor Y_i \rfloor - 1$.
On the other hand, if no such slot exists, we keep $p' = p$,
that is, we do not change the slot pointed to by the proxy.
Now, we process $p'$ as a regular barely open fractional
slot of $y = y''_p + Y_i - \lfloor Y_i \rfloor - 1$. $p'$
is processed in the same manner as a fractionally open slot
$t_{d_i} - \lfloor Y_i \rfloor -1$ in the
right-shifted solution, even if the proxy points to some other slot $p'$.
If the fractional slot $p'$ gets opened ($y'_{p'}$
becomes 1 by the rounding process), then we do not pass over any proxy.
Otherwise, a proxy of cost $y''_{p'}$ gets carried over
to iteration $i+1$, pointing to $p'$.
By the above proxy merging procedure, there can be at most one proxy in an iteration.

\subsubsection{Processing $Y_i$.}
In the following discussion, we assume $Y_i$
\emph{already takes into account
any proxy from iteration $i-1$} as described above.

\begin{customcase}{1}
$Y_i -\lfloor Y_i \rfloor = 0$.
\end{customcase}

In this case, slots $\{t_{d_i} - Y_i+1,\ldots, t_{d_i}\}$ are fully open.
We set $y'_t= y_t$ for all $t \in \{t_{d_i} - Y_i+1,\ldots, t_{d_i}\}$,
and add slots $\{t_{d_i} - Y_i+1, \ldots, t_{d_i}\}$ to $\mathcal{O}_i$.

\begin{customcase}{2}
$Y_i >1$ and $Y_i - \lfloor Y_i \rfloor \geq \frac{1}{2}$.
\end{customcase}
Here, slots $\{t_{d_i} - \lfloor Y_i \rfloor+1, \ldots, t_{d_i}\}$
are fully open and the slot $t_{d_i} - \lfloor Y_i \rfloor $
is half-open.  We set $y'_t = y_t$
for $t \in \{t_{d_i} - \lfloor Y_i \rfloor+1, \ldots, t_{d_i}\}$,
and $y'_ {t_{d_i}-\lfloor Y_i \rfloor} = 1$.
We charge the cost of $y'_{t_{d_i}-\lfloor Y_i \rfloor}$ to
$y_{t_{d_i}-\lfloor Y_i \rfloor}$, charging
$y_{t_{d_i}-\lfloor Y_i \rfloor}$ at most twice in the process.
We add the slots $\{t_{d_i} - \lfloor Y_i \rfloor, \ldots, t_{d_i}\}$
to $\mathcal{O}_i$.

\begin{customcase}{3}
$Y_i >1 $ and $Y_i - \lfloor Y_i \rfloor < \frac{1}{2}$.
\end{customcase}

In this case, slots $\{t_{d_i} - \lfloor Y_i \rfloor+1, \ldots, t_{d_i}\}$
are fully open and the slot $t_{d_i}-\lfloor Y_i \rfloor$ is barely open.
%
We first close $t_{d_i}-\lfloor Y_i \rfloor$,
and check if a feasible assignment
of all jobs in $\bigcup_{j\leq i}{\textscr{J}_j}$ exists in the
slots $\bigcup_{k \in [i-1]}{ \mathcal{O}_k} \cup 
\{t_{d_i} - \lfloor Y_i \rfloor+1, \ldots, t_{d_i}\}$
using the maximum-flow construction described earlier.
If successful, we add
$\{t_{d_i} -\lfloor Y_i \rfloor+1, \ldots, t_{d_i}\}$
to $\mathcal{O}_i$, and pass on a proxy of of cost
$Y_i - \lfloor Y_i \rfloor$ and move to
the next deadline (iteration $i+1$). The pointer to this proxy would be
the slot $t_{d_i}-\lfloor Y_i \rfloor$ if this is not coincident
with $t_{d_{i-1}}$, otherwise to an earlier slot
pointed to by a proxy coming from iteration $(i-1)$, as described earlier.

Otherwise, if no feasible assignment is
found by the maximum-flow procedure,
we need to open the barely open slot $t_{d_i} - \lfloor Y_i \rfloor$,
and account for its cost.
We add $t_{d_i} -\lfloor Y_i \rfloor$ to $\mathcal{O}_i$
over and above the slots
$\{t_{d_i} -\lfloor Y_i \rfloor+1, \ldots, t_{d_i}\}$.
We charge $t_{d_i} -\lfloor Y_i \rfloor$ as a
\emph{dependent} on the \emph{earliest} fully
open slot that has no dependents\footnote{Note that here fully open slots
denote the slots $t$, where $y_t = 1$ either in the original
LP solution or after processing for proxy slot from earlier iteration.
In other words, fully open slots are those that are open in the
rounded solution, charging their costs to \emph{themselves}, and none else.
The set of fully open slots may not be the same as
$\bigcup_{k \in [i]}{\mathcal{O}_k}$ in iteration $i$.}.

Suppose all earlier fully open slots have dependents or are
parts of trios, then, we charge the cost of opening the slot
$t_{d_i} - \lfloor Y_i \rfloor$ to $t_{d_i} - \lfloor Y_i \rfloor + 1$. This
is feasible, since $Y_i > 1$, and therefore, $t_{d_i} - \lfloor Y_i \rfloor + 1$ is fully open.
Moreover, the slot $t_{d_i} - \lfloor Y_i \rfloor + 1$
cannot have any dependents from earlier,
because the rounding process in any iteration $k$
allows charging fully open slots
in time slots equal to or before $t_{d_{k}}$.

\begin{customcase}{1}
$1>Y_i\geq \frac{1}{2}$.
\end{customcase}

In this case, we open the slot $t_{d_i}$ integrally, charging the cost
$Y_i$ at most twice. Also, since this already takes
any proxy coming from earlier iteration into account,
no proxy is passed over from this iteration.

\begin{customcase}{2}
$Y_i<\frac{1}{2}$.
\end{customcase}

We will first try to close $t_{d_i}$.
We check if a feasible assignment
of all jobs in $\bigcup_{j\leq i}{\textscr{J}_j}$ exists in the slots
$\bigcup_{k \in [i-1]}{\textscr{O}_k}$,
using the max-flow construction described earlier.
If such an assignment exists, we keep $t_{d_i}$ closed
and move on to the next deadline,
passing over a proxy of cost $Y_i = y_{t_{d_i}}$ to
the iteration $i+1$, pointing to the slot $t_{d_i}$.
Note that this takes into account any
proxy from the previous iteration $i-1$, since
the $Y_i$ already takes proxy into account and the
pointer is set to $t_{d_i}$ without loss of generality
as described earlier.

Suppose closing $t_{d_i}$ and finding a feasible assignment of jobs in
$\bigcup_{j\leq i}{\textscr{J}_j}$ is not successful.
Then we are forced to open $t_{d_i}$, but we need to account for it.
We need to charge it either to an  earlier fully open slot either
as a dependent or a trio, or to an earlier half-open slot as a filler.
We find the earliest fully open slot that does not have a dependent
and charge to it. If all the fully open earlier slots have dependents,
we find the earliest fully open slot with a dependent with which it can form a trio,
that is, the cumulative $y$ cost of the fully open slot along with its dependent and
the current barely open slot $t_{d_i}$ is at least $\frac{3}{2}$.
If all the earlier fully open slots are already parts of trios, or if not, no trio
is possible with them and their dependents, then we find the earliest
half-open slot to charge it as a filler. This is possible only if
the cumulative sum of the $y$ cost of the half-open slot and $t_{d_i}$
is at least $1$.
This completes the description of the rounding process.

We will now show that after we process deadline $t_{d_i}$
by the rounding algorithm,
there will exist feasible assignment of all jobs with deadline $\leq t_{d_i}$
in the integrally opened slots up to $t_{d_i}$.
After this, we will show that we will always be able to find
a way to charge such a barely open slot that needs to be opened by the
rounding algorithm, given a feasible LP solution.

\begin{lemma}
\label{lemma:feasible}
There exists a feasible assignment of all jobs in $\bigcup_{x\in [1, \ldots, i]}{\mathcal{J}_x}$
in the set of integrally opened slots $\bigcup_{x\in [1, \ldots, i]}{\mathcal{O}_x}$
after we process deadline $t_{d_i}$ for some $i>1$, provided that there exists a feasible assignment
of all jobs in $\bigcup_{x\in [1, \ldots, i-1]}{\mathcal{J}_x}$
in the set of integrally opened slots $\bigcup_{x\in [1, \ldots, i-1]}{\mathcal{O}_x}$
after processing deadline $t_{d_{i-1}}$.
\end{lemma}

\begin{proof}
Consider the iteration when we are processing deadline $t_{d_i}$.
If $Y_i$ is $0$, no processing is required and we trivially satisfy the claim in the
Lemma. Hence, consider $Y_i>0$.
First consider the case when $0<Y_i\leq \frac{1}{2}$ and we close
the barely open slot $t_{d_i}$. In this case, clearly there exists a feasible
assignment of all jobs in $\bigcup_{x\in [1, \ldots, i]}{\mathcal{J}_x}$
in the set of integrally opened slots $\bigcup_{x\in [1, \ldots, i-1]}{\mathcal{O}_x}$,
since we close the slot $t_{d_i}$ only if flow is able to find a feasible
assignment of all jobs in $\bigcup_{x\in [1, \ldots, i]}{\mathcal{J}_x}$ in the
set of integrally opened slots thus far.
Next, consider the case when $0<|\lceil Y_i\rceil - Y_i|<\frac{1}{2}$,
and we open $\lfloor Y_i \rfloor$ slots, closing the barely open slot
$t_{d_i} - \lfloor Y_i \rfloor$. Again, clearly there exists a feasible
assignment of all jobs in $\bigcup_{x\in [1, \ldots, i]}{\mathcal{J}_x}$
in the set of integrally opened slots $\bigcup_{x\in [1, \ldots, i-1]}{\mathcal{O}_x}$,
since we close the slot $t_{d_i} - \lfloor Y_i \rfloor$ only if flow is able to find a feasible
assignment of all jobs in $\bigcup_{x\in [1, \ldots, i]}{\mathcal{J}_x}$ in the
set of integrally opened slots thus far.
The above cases correspond to the scenario when the rounding algorithm opens $\lfloor Y_i \rfloor$
slots integrally. Now, let us consider the cases when the rounding algorithm opens
$\lceil Y_i \rceil$ slots integrally. In this case, clearly, there would
be no proxy passed over from iteration $i$.
We will integrally open slots $t_{d_i} - \lfloor Y_i \rfloor, \ldots, t_{d_i}$.
We know that there exists a feasible assignment of $\bigcup_{x\in [i-1]}{\mathcal{J}_x}$
in integrally open slots $\bigcup_{x\in [1, \ldots, i-1]}{\mathcal{O}_x}$.
Hence, considering only the jobs $\bigcup_{x \in [i]}{\mathcal{J}_x}$,
a feasible \emph{fractional} LP solution would exist on the slots $\sum_{x \in [i-1]}{|\mathcal{O}_x}$,
and on the slots $[t_{d_i} - \lfloor (Y_i + y''_{p_{i-1}})\rfloor, \ldots, t_{d_i}]$,
(where we have explicitly indicated that the proxy from iteration $i-1$ gets
added to $Y_i$, when it gets processed). Recall from Section \ref{sec:proxy},
that a proxy gets passed over only when a feasible assignment of jobs
exist in the integrally opened slots without considering the cost of the proxy slot
(corresponding to a barely open slot). Hence, a feasible LP solution could
open the integrally open slots up to iteration $i-1$, as fully open, the slot
$t_{d_i} - \lfloor (Y_i + y''_{p_{i-1}})$ fractionally up to
$Y_i + y''_{p_{i-1}} - \lfloor (Y_i + y''_{p_{i-1}})\rfloor$,
and the slots $[t_{d_i} - \lfloor (Y_i + y''_{p_{i-1}})\rfloor + 1, \ldots, t_{d_i}]$, as fully open.
In case $t_{d_{i-1}} = t_{d_i} - \lfloor (Y_i + y''_{p_{i-1}})\rfloor$, we open
the actual slot pointed to by the proxy slot: $p_{i-1}$ up to
$Y_i + y''_{p_{i-1}} - \lfloor (Y_i + y''_{p_{i-1}})\rfloor$,
and by the property of the rounding algorithm as explained
in Section \ref{sec:proxy}, such an unopened slot will exist for
a proxy with non-zero cost.
Clearly, such an LP solution would still be feasible if we now open the only fractionally
slot (either slot $t_{d_i} - \lfloor (Y_i + y''_{p_{i-1}})\rfloor$, or, the actual slot pointed
to by the proxy $p_{i-1}$) fully as well. However, now we get a feasible fractional assignment
of all jobs in $\bigcup_{x \in [i]}{\mathcal{J}_x}$ in the integrally open
slots $\bigcup_{x\in [1, \ldots, i-1]}{\mathcal{O}_x} \cup [t_{d_i} - \lfloor (Y_i + p_{i-1})\rfloor, \ldots, t_{d_i}]$.
Since $\mathcal{O}_i = [t_{d_i} - \lfloor (Y_i + p_{i-1})\rfloor, \ldots, t_{d_i}]$,
therefore, by integrality of flow, there would exist a feasible integral assignment of
all jobs in $\bigcup_{x \in [i]}{\mathcal{J}_x}$ in $\bigcup_{x \in [i]}{\mathcal{O}_x}$.
This completes the proof.
\qed
\end{proof}


\begin{lemma}
\label{lemma:halfopen}
If the rounding process decides to open a barely open slot $t_{d_i} - t$, $t\geq 0$ in an iteration $i$,
then we will always find a fully open slot to charge it as a dependent or as a trio, or else,
a half-open slot to charge it as a filler.
\end{lemma}

\begin{proof}
If the rounding process opens a barely open slot $t_{d_i}-t$, $t> 0$, then that would imply
$Y_i> 1$ because of the right-shifted LP solution structure. (Specifically, a slot $t_{d_i}-t$, with $t>0$
can be barely open only if $Y_i> 1$.) In such a case, we can always find a fully open
slot $\leq t_{d_i}$ to charge it to, if not a slot earlier. This is because, as already argued,
$t_{d_i}$ must be fully open, and because of the sequential processing of deadlines in increasing
order by the rounding algorithm, no barely-open slots from the previous iterations could have charged it.

Now, let us consider the case when the rounding opens a barely-open slot $t_{d_i}$ (this implies $Y_i < \frac{1}{2}$).
We assume that the rounding process was feasible till the iteration $i-1$, without loss
of generality as argued earlier.
In iteration $i$, for contradiction, assume that all the fully open slots $t< t_{d_i}$
in the right-shifted LP solution have dependents or are parts of trios, or no trio is
possible with the current dependent of any fully open slot, since the sum of the $y$s
are not sufficient, and no filler is possible with any earlier half-open slot (either because
they already have fillers, or because the sum of the $y$s is not sufficient).
We show some structural properties of the right shifted LP solution, which we will use to
derive the contradiction.

The first property we show is as follows: the last fully open slot occurring earlier to $t_{d_i}$ corresponds to a deadline.

\begin{myclaim}
\label{claim:deadline}
Let $t_{max}$ denote the latest fully open slot in the right-shifted LP solution,
that occurs before $t_{d_i}$. More formally, $t_{max} = \arg\max{t|t< t_{d_i} \cap y_{t} = 1}$.
Then $t_{max}$ must correspond to a deadline.
\end{myclaim}
\begin{proof}
Suppose the above claim is false, and $t_{max}$ does
not correspond to any deadline. Let the deadline immediately after $t_{max}$ be $t_{d_k}$;
therefore, $t_{d_k} > t_{max}$ and $t_{d_{k-1}} < t_{max}$. Since $t_{max}$ is the latest fully open
slot occurring before $t_{d_i}$, $t_{d_{k}}$ must be either closed or barely open or half open.
However, from Observation \ref{obs:rightshifted}, in the
right-shifted solution structure, if $t_{max}$ is
fully open, $t_{d_k} > t_{max}$ must be fully open. Therefore, this proves the above claim
by contradiction.
\qed
\end{proof}

Let $t_{max}$ correspond to the $k^{th}$ deadline, i.e., $t_{max} = t_{d_k}$, where $k<i$.

The next property we show is as follows: any slot in $t' \in [t_{d_k} + 1, \ldots t_{d_i} - 1]$
such that $y_{t'} >0$ must correspond to a deadline.
\begin{myclaim}
\label{claim:deadline2}
Any slot $t' \in [t_{d_{k} +1}, \ldots, t_{d_i}]$ such that $y_{t'}>0$
in the right-shifted LP solution must correspond to a deadline, i.e., $t'= t_{d_j}$
for some $j \in \textscr{J}$.
\end{myclaim}

\begin{proof}
Suppose $t'$ does not correspond to any deadline. We know that $t_{d_k}$ is the first fully open
slot going backwards from $t_{d_i}$. Let the first deadline after $t'$ be $t_{d_j}$
where $t_{d_j} \leq t_{d_i}$. Clearly, $t_{d_j}$ is well-defined and exists.
Moreover, $t_{d_j}$ is either closed, or barely open or half-open. However, from 
Observation \ref{obs:rightshifted}, if $y_{t'}> 0$, then $y_{t_{d_j}}= 1$. Since that is
not true, this completes the proof by contradiction.
\qed
\end{proof}


Next, we argue that if any slot in $t \in [t_{d_k} + 1, \ldots t_{d_i} - 1]$
(with $y_{t} >0$) was opened by the rounding
process, i.e., $y'_t = 1$, then it must be half-open with which $t_{d_i}$ can form a filler.
\begin{myclaim}
\label{claim:filler}
Let $t$ be the last slot $\in [t_{d_k} + 1, \ldots t_{d_i} - 1]$,
with $y_{t}>0$, that was opened by the rounding
process in an earlier iteration, i.e., $y'_t = 1$. Then $t$ must be a half-open deadline
such that $y_{t} + y_{t_{d_i}} \geq 1$, hence $t_{d_i}$
can be opened by charging $t$ as a filler at
most twice the cost of $y_{t} + y_{t_{d_i}}$.
\end{myclaim}

\begin{proof}
Suppose for the sake of contradiction $t \in [t_{d_k} + 1, \ldots t_{d_i} - 1]$, the
last slot with $y_{t}>0$, is either barely open, or half-open, such that
$y_t + y_{t_{d_i}} <1$ (i.e., filler is not possible), and has been opened
by the rounding process in an earlier iteration, i.e., $y'_t = 1$.

We know that this $t$ is a deadline from Claim \ref{claim:deadline2}, and also,
it cannot be fully open, since by assumption, $t_{d_k}$ is the last fully open
deadline occurring before $t_{d_i}$.
Let $t$ correspond to the $p^{th}$ deadline, i.e, $t=t_{d_p}$, where $k<p<i$.

We have assumed that the rounding is feasible till iteration $i-1$,
at at most twice the cost of the LP solution up to deadline $i-1$.
Therefore, for opening $t_{d_p}$, for $p<i$,
$t_{d_p}$ must have feasibly charged an earlier fully open slot (in case $t_{d_p}$ was barely open)
or itself, if it was half-open.

Since $t_{d_i}$ itself is barely open, the jobs in $\textscr{J}_i$ must be feasible in
$t_{d_{p}}$, as the LP solution would have had to schedule some portion
of all the jobs scheduled in $t_{d_i}$ in $t_{d_p}$ or earlier, due to the feasibility constraint
$x_{i,j} \leq y_i$. Therefore, the job assignments can be shifted as is from $t_{d_i}$ to $t_{d_p}$,
after increasing $y_{d_p}$ to $y_{d_p} + y_{d_i}$, without violating LP feasibility
(by assumption $y_{d_p} + y_{d_i} < 1$), or without increasing the
LP cost. However, that implies that there exists a feasible, equivalent LP solution where
the slot $t_{d_i}$ is closed, and $t_{d_p}$ is either half-open or barely open.
In the former case, $t_{d_i}$ would charge itself and in the latter case, $t_{d_i}$
would continue to be a dependent\/trio\/filler on the slot it was already charging.
From the argument in Lemma \ref{lemma:feasible}, we can see that
this implies that there exists a feasible assignment of
all jobs in $\bigcup_{a \in [i]}{\mathcal{J}_a}$
in the set of integrally open slots
$\bigcup_{b \in [p]}{\mathcal{J}_b}$, where $p<i$.
In other words, there exists a feasible fractional (hence, integral) flow of an amount equal to the
cumulative size of all jobs in
$\bigcup_{a \in [i]}{\mathcal{J}_a}$ in
integrally open slots $\bigcup_{b \in [p]}{\mathcal{O}_b}$.
Hence, this contradicts the assumption that
flow could not find a feasible assignment
of all the jobs in $\bigcup_{a \in [i]}{\mathcal{J}_a}$
in the current set of
integrally open slots $\bigcup_{b \in [p]}{\mathcal{O}_b}$,
which is why the LP had to open $t_{d_i}$ in the
first place and hence charge it somewhere.
Therefore, we have proved the claim by contradiction.
\qed
\end{proof}

From Claim \ref{claim:filler},
it is clear, that if there is any slot $t\in [t_{d_k} + 1, \ldots t_{d_i} - 1]$
with $y_{t}>0$ that has been opened by the rounding process already,
then the last such slot must be half-open with which $t_{d_i}$ can form a filler.
Therefore, if such a slot exists, then in iteration $i$, if
the rounding algorithm needs to open $t_{d_i}$
integrally, we can charge it at most twice the cost
to $t_{d_p}$ feasibly. This is a contradiction to the assumption in Lemma \ref{lemma:halfopen}.
Henceforth, we assume that no slot $t \in [t_{d_k} + 1, \ldots t_{d_i} - 1]$
has been opened by the rounding algorithm.

Now, we only need to consider the cases where the rounding algorithm
needs to open $t_{d_i}$ and we are unable to charge it to any fully open slot
as dependent or trio, where the fully open slot must be $t_{d_k}$ or earlier.
However, this implies that $t_{d_k}$ must have a dependent, or be a part of a trio.

Let us consider the case when $t_{d_k}$ is a part of a trio.
A trio can happen only when a fully open slot is charged by two barely open
slots, one occurring before it, and one occurring after it.
However, that would mean that there is some slot with positive $y$
between $t \in [t_{d_k} + 1, \ldots t_{d_i} - 1]$ that has been opened
by the rounding algorithm, that is not possible.

Hence, we only need to consider the case when $t_{d_k}$ has a dependent
charging it, where the dependent occurs earlier than $t_{d_k}$.
However, this means that the dependent cannot be processed in any iteration
earlier than the iteration in which deadline $t_{d_k}$ is processed
by the definition of the rounding process. At the same time, it must hold that
all fully open slots occurring earlier than $t_{d_k}$ must have dependents with which trios are not
possible.

Therefore, either the dependent is the barely open slot $t = t_{d_k} -1$
(due to the right-shifted LP solution structure) and clearly, $Y_k < 2$,
or it is a proxy slot coming from earlier iterations.
However, if it is a proxy, then that means an earlier barely open slot
$t'$, was closed without any loss of feasibility, where $t_{d_{j-1}}+1\leq t'\leq t_{d_j}$, for some $j<k$.
It also means that no barely open or half open slots could have opened between
$t_{d_j}$ and $t_{d_k}$ as otherwise it would
have absorbed the proxy.
Moreover, there are no fully open slots between $t_{d_j}$ (inclusive of $t_{d_j}$) and $t_{d_k}$
since the proxy would have charge this earlier slot instead of $t_{d_k}$\footnote{If
there was a fully open slot, it would have remained uncharged so far
since no barely open slot has opened from iteration $j$ till iteration $k$.}.
Therefore, $t_{d_k}$ must be the first fully open slot from
$t_{d_j}$ onwards. It also implies that all the jobs in
$\bigcup_{x\in \{1, k-1\}}{\mathcal{J}_x}$ do not
need the proxy value for a feasible assignment. Hence, we can change the pointer 
of the proxy slot to $t_{d_k}-1$ without any loss of generality and consider $t_{d_k}-1$
as dependent on $t_{d_k}$. (Note that $t_{d_k}-1$ may also be equal to $t_{d_j}$, in which case
we do not need to change anything.)
Hence, without loss of generality, we can consider the dependent on $t_{d_k}$ to be
the barely open slot $t_{d_k}-1$.

We next argue that the above case is not possible, in other words, flow
will find a feasible assignment for all jobs with deadlines $\leq t_{d_i}$,
even after closing $t_{d_i}$.
We know by induction hypothesis, that all jobs with deadline $\leq t_{d_{k}}$
have a feasible assignment in the integrally open slots up to slot $t_{d_k}$.
Consider an optimal packing of the jobs in the integrally open slots.
Now, flow could not find a feasible assignment, hence there is at least one
job with deadline $t_{d_i}$ that could not be accommodated in the fully open
slot $t_{d_k}$, that must in turn be fully packed. Either none of these jobs
could be moved earlier due to release time constraints, otherwise,
slot $t_{d_k}-1$ is also full. We continue moving backwards in this manner, traversing
through the fully packed integrally open slots,
till we finally come across a pair of adjacent integrally open slots, $t$ and $t'$,
$t'< t$, where all integrally open slots going backwards between $t_{d_k}$
and $t$ (both inclusive) are completely packed, while $t'$ has space, however none of the jobs
assigned in slots $\{t, \ldots, t_{d_k}\}$, including those with deadline $t_{d_i}$,
can be moved any earlier due to release time constraints.
This clearly implies that the LP solution, being feasible, would have to schedule
this set of jobs in slots also in the same time range $[t, \ldots, t_{d_i}]$.

Let there be $N$ integrally open slots in $[t, \ldots t_{d_k}]$, including slots 
$t_{d_k}$ and $t_{d_k}-1$ (we have argued that $t_{d_k}-1$ must be a
barely open slot dependent on $t_{d_k}$, that was opened integrally by the rounding
algorithm). The LP solution would have therefore scheduled $N$ job units
in the time range $[t, \ldots, t_{d_k}]$.
Now, in the $N-2$ integrally open slots occurring before $t_{d_k}-1$,
let there be $x_1$ fully open slots with dependents (with which trios
were not possible) and $x_2$ be fully open slots that are part of trios.
Since all fully open slots were already charged (as $t_{d_k}-1$
had to charge $t_{d_k}$), there are no other fully open slots.
Furthermore, this implies that there are $x_1$ barely open slots dependent on $x_1$
fully open slots and $2 x_2$ barely open slots forming trios with the $x_2$
fully open slots. The remaining slots must be half open. Let $x_3$ of the
half open slots have fillers (with $x_3$ barely open slots) and $x_4$ of
the remaining half-open slots.
Clearly, $2 x_1 + 3 x_3 + 2 x_3 + x_4 + 2 = N$.
Each of the dependent-fully open pair contribute $< \frac{3}{4}$ on an average
to the LP solution cost, the trios contribute $< \frac{2}{3}$
on an average, the half-open slots with fillers contribute $< \frac{3}{4}$ on
an average, and the remaining half-open slots contribute $< 1$ each.
Finally, $t_{d_k}-1$, $t_{d_k}$ and $t_{d_i}$ together contribute $< \frac{3}{2}$
since they cannot form trio.
Therefore, the total LP solution cost up to $t_{d_i}$ is
$< 2 x_1 \frac{3}{4} + 3 x_3 \frac{2}{3} + 2 x_3 \frac{3}{4} + x_4 + \frac{3}{2} < N$.
This gives a contradiction. Hence, this case too cannot arise, and
a barely open slot will always find a fully open slot or half-open slot to charge,
if flow cannot find an assignment of all jobs up to the current deadline being processed.

The proof of Lemma~\ref{lemma:halfopen} therefore follows.
\qed
\end{proof}

%
The next theorem, proving the approximation guarantee of the LP rounding algorithm
follows from Lemma \ref{lemma:feasible} and Lemma \ref{lemma:halfopen}.

\begin{theorem}
There exists a polynomial time algorithm which gives a solution of cost
at most twice that of any optimal solution to the active time problem on non-unit
length jobs with integral preemption.
\end{theorem}
\begin{proof}
From Lemma \ref{lemma:halfopen}, it follows that at the end of every iteration $i$,
the number of integrally open slots is at most twice the cost of the LP solution 
up to $t_{d_i}$.
Formally, at the end of every iteration $i$,
$|\bigcup_{x\in [i]}{\mathcal{O}_x}| \leq 2 \sum_{x \in [i]}{Y_i}$, and
from Lemma \ref{lemma:feasible}, it follows by induction that there exists a
feasible integral assignment of jobs in $\bigcup_{x\leq i}{\textscr{J}_x}$
in $\bigcup_{x\in [i]}{\mathcal{O}_x}$. The base case is
given by Lemma \ref{lemma:basecase}.
We repeat this until the last deadline $d_{\ell}$.
At the end, by induction we are assured of an integral feasible
assignment on the set of opened slots via maximum flow,
while the number of open slots is at most twice the
optimal LP objective function value.
Hence, we get a $2$-approximation.
The time complexity of the rounding algorithm is $O(\ell T F)$, where $F$ is the complexity
of the maximum flow algorithm and $\ell$ is the number of distinct deadlines.
The preprocessing is linear in $T$ and the LP is polynomial in $T$, where $T$ is the
number of distinct time slots in the union of the feasible time intervals for
all the jobs in the instance.
\qed
\end{proof}

\subsection{LP Integrality Gap}
We show here that the natural LP for this problem
has an integrality gap of two.
Hence, a $2$-approximation is the best possible using LP rounding.
Consider, $g$ pairs of adjacent slots.
In each pair, there are $g+1$ jobs which can
only be assigned to that pair of slots.
An integral optimal solution will have
cost $2g$, where as in an optimal fractional solution,
each such pair will be opened up to $1$ and $\frac{1}{g}$,
and all the $g+1$ jobs will be assigned up to
$\frac{g}{g+1}$ to the fully open slot, and up to $\frac{1}{g+1}$,
to the barely open slot, thus maintaining all the constraints.
Therefore, optimal LP solution has cost $g+1$ and
$\frac{g+1}{2g} \rightarrow 2$ as $g\rightarrow \infty$.

\section{Busy Time}

\subsection{Notation and Preliminaries}
In the busy time problem, there are
an unbounded number of machines to
which jobs can be assigned, with each machine limited
to working on at most $g$ jobs at any given time $t$.
Informally, the busy time of a single machine is 
the total time it spends working on at least one assigned job.
Unlike in the active time model, time is not slotted and 
job release times, deadlines
and start times may take on real values.
The goal is to feasibly assign jobs to machines to
minimize the schedule's busy time, that is,
the sum of busy times over all machines.
If $p_j < d_j-s_j$, then start times should also be specified.
However, for the following insightful special case, 
job start times are determined by their release times.

\begin{definition}
A job $j$ is said to be an \emph{interval job} when 
$p_j = d_j - r_j$. 
\end{definition}

The special case of interval jobs is
central to understanding the general problem.
In this section, we give
improvements for the busy time problem 
via new insights for the interval job case.
In general, we will let $\scr{J}'$ 
refer to an instance of jobs that are not necessarily interval,
and $\scr{J}$ to an instance of interval jobs.
Job $J_j$ is \emph{active} on machine 
$m$ at some time $t\in [r_j,d_j)$ if $J_j$ is 
being processed by machine $m$ at time $t$. 

\begin{definition}
The length of time interval $I = [a, b)$ is 
$\ell(I) = b - a$. The span of $I$ is also $Sp(I) = b-a$.
\end{definition}

We generalize these definitions to sets of intervals.
The span of a set of intervals
is informally the magnitude of the projection
onto the time axis and is at most its length.
Sometimes we refer to
$\ell(\scr{S})$ as the ``mass'' of the set $\scr{S}$. 

\begin{definition}
For a set $\scr{S}$ of interval jobs, its length
is $\ell(\scr{S}) = \sum_{I \in \scr{S}} \ell(I)$.
For two interval jobs $I$ and $I'$, the span of $\scr{S}=\{I,I'\}$
is defined as $Sp(\scr{S}) = \ell(I) + Sp(I') - \ell(I \cap I')$.
For general sets $\scr{S}$ of interval jobs, 
$Sp(\scr{S}) = \ell(I_1) + Sp(\scr{S} \backslash I_1) - 
\ell(I_1 \cap (\scr{S} \backslash I_1))$, where $I_1$ is
the job in $\scr{S}$ with earliest release time.
\end{definition}
%
%

We need to find a partition of the jobs into 
groups or \textit{bundles}, so that every bundle has
at most $g$ jobs active at any time $t$. 
Each bundle $\scr{B}$ is assigned to its own machine,  
with busy time $Sp(\textscr{B})$. 
Suppose we have partitioned the job set into 
$\kappa$ feasible bundles (the feasibility respects the 
parallelism bound $g$ as well as the release 
times and deadlines).  Then the total busy time of 
the solution is $\sum^\kappa_{k=1} Sp(\scr{B}_k)$. 
The goal is to minimize this quantity.
We consider two problem variants: 
$g$ bounded and $g$ unbounded.
For the \emph{preemptive} version of the problem, 
the problem definition remains the same, the only 
difference being that the jobs can be processed 
preemptively across various machines. 

If the jobs are not necessarily interval jobs,
then the difficulty of finding the minimum busy time
lies not just in finding a good partition of jobs, 
but also in deciding when each job should start.
We study both the preemptive and non-preemptive
versions of this problem.

Without loss of generality, the busy 
time of a machine is contiguous. 
If it is not, we can break it 
up into disjoint periods of contiguous 
busy time, assigning each of them to 
different machines, 
without increasing the total busy time of the solution. 

Let $OPT(\scr{J}')$ be the optimal busy time 
of an instance $\scr{J}'$,
and $OPT_{\infty}(\scr{J}')$ 
the optimal busy time when unbounded 
parallelism is allowed. 
The next lower bounds on any optimal solution for a 
given instance $\scr{J}'$ were introduced earlier 
(\cite{ESA03}, \cite{FSTTCS05}). 
The following ``mass'' lower bound 
follows from the fact that on any machine, there are at
most $g$ simultaneously active jobs.

\begin{observation}
\label{obs0}
$OPT(\scr{J}') \geq \frac{\ell(\scr{J}')}{g}$. 
\end{observation}

The following ``span'' lower bound follows from the fact that 
any solution for bounded $g$ is also a feasible
solution when the bounded parallelism constraint is removed.
If all jobs in $\textscr{J}$ are interval jobs, then
$OPT_\infty(\textscr{J}) = Sp(\textscr{J})$.  

\begin{observation}
\label{obs:spinfty}
$OPT(\scr{J}')\geq OPT_\infty(\scr{J}')$.
\end{observation}

However, the above lower bounds individually can be arbitrarily bad. 
For example, consider an instance of $g$ disjoint unit 
length interval jobs. 
The mass bound would simply give a lower bound of $1$, whereas the optimal 
solution pays $g$. Similarly, consider an instance 
of $g^2$ identical unit length 
interval jobs. The span bound would give a 
lower bound of $1$, whereas the optimal 
solution has to open up $g$ machines for unit intervals, paying $g$. 

We introduce a stronger lower bound, which we call 
the demand profile. In fact, the algorithm of Alicherry 
and Bhatia \cite{ESA03} as well as that of Kumar and 
Rudra \cite{FSTTCS05} implicitly charge the demand profile. 
This lower bound holds for the case of interval jobs. 

\begin{definition}
\label{def:dp}
Let $A(t)$ be the set of interval jobs that are active at time $t$,
i.e., $A(t) = \{j: t\in [r_j,d_j)\}$. 
Also, let $|A(t)|$ be the \emph{raw demand} at time $t$, 
and $D(t) = \left\lceil\frac{|A(t)|}{g}\right\rceil$
the \emph{demand} at time $t$.
\end{definition}

\begin{definition}
An interval $I = [a,b)$ is \emph{interesting} if 
no jobs begin or end within $I$, and 
$\min_j r_j \le a \le b \le \max_j d_j$.
\end{definition}

For a given instance, there are at most $2n$ interesting
intervals.
Also, the raw demand, and hence the demand, 
is uniform over an interesting interval.
Thus, it makes sense to talk about the demand over such an
interval.  Let $A(I_i)$ ($D(I_i)$, respectively) 
denote the raw demand (demand, resp.) over 
interesting interval $I_i$.

Then for a set $\textscr{I}$ of interesting intervals, 
$\textscr{I} = \{I_1, I_2, \ldots, I_\ell\}$, 
$\ell\leq 2n$, we have that $D(I_i) = D(t), \ \forall t \in I_i$. 
Additionally, 
$Sp(\textscr{J}) = Sp(\bigcup_{I_i \in \textscr{I}} I_i)$.

\begin{definition}
For an instance $\textscr{J}$ of interval jobs, 
its demand profile $DeP(\textscr{J})$ is
the set of tuples $\{(I_i,D(I_i))\}_{I_i \in \textscr{I}}$.
\end{definition}

The demand profile is expressed in terms of $O(n)$ tuples, 
regardless of whether or not release times, deadlines, 
or job lengths are polynomial in $n$.
%
The demand profile of an instance yields a lower bound
on the optimal busy time.
We can think of the cost of an instance $\textscr{J}$'s
demand profile as 
$\sum_{I_i \in \textscr{I}}{D(I_i)}$.  

\begin{observation}
\label{obs:dp}
For an instance $\textscr{J}$ of interval jobs,
$OPT(\textscr{J}) \geq \sum_{I_i \in \textscr{I}}{D(I_i)}$,
where $\textscr{I}$ is the set of interesting intervals taken
with respect to $\textscr{J}$. 
\end{observation}

\begin{proof}
There are $|A(I_i)|$ active jobs within an interesting 
interval $I_i$. Then any feasible solution has
$\left\lceil \frac{|A(I_i)|}{g}\right\rceil$ machines busy 
during the interval $I_i$. Moreover, 
$Sp(\textscr{J}) = Sp(\textscr{I})$.
\qed
\end{proof}

\subsection{A $2$-approximation for busy time with interval jobs}

In this section, we briefly outline how the work of 
Alicherry and Bhatia \cite{ESA03} and that of 
Kumar and Rudra \cite{FSTTCS05} for related problems 
imply $2$-approximations for 
the busy time problem for interval jobs, 
improving the best known factor of four \cite{IPDPS09}. 
A detailed description can be found in Appendix~\ref{app:intjobs}.

Both the above mentioned works consider request routing problems 
on interval graphs, motivated by optical design systems. 
The requests need to use links on the graphs for being 
routed from source to destination, 
and the number of requests that can use a link is bounded. 
The polynomial time complexity of the algorithms  
crucially depends on the fact that the request (job)
lengths are linear in the number of time slots; 
this does not hold for the busy time problem 
where release times, deadlines and 
processing lengths may be real numbers.
However, even if release times and deadlines 
of jobs are not integral, there can be at most 
$2n$ \emph{interesting} intervals, 
such that no jobs begin or end within the interval. 
The demand profile is uniform over every interesting interval. 
Therefore, their algorithms can be applied to the 
busy time problem with this simple modification, 
thus maintaining polynomial complexity. 
In order to bound the performance 
of their algorithms for the busy time problem, 
we additionally need to assume that the 
demand everywhere is a multiple of $g$. 
However, for an arbitrary instance,
we can add dummy jobs spanning any interesting 
interval $I_i$ where the raw demand 
$|A(I_i)|$ is not a multiple of $g$ without changing the demand profile. 
Specifically, if $cg< |A(I_i)| \le (c+1)g$ for some $c\geq 0$, 
then $DeP(I_i) = c+1$ and adding $(c+1)g - |A(I_i)|$ 
jobs spanning $I_i$  does not change the demand profile. 
Hence, applying their algorithms to a suitably 
modified busy time instance of interval jobs, 
will cost at most twice the demand profile. 

\begin{theorem}
There exist $2$-approximation polynomial time 
algorithms for the busy time problem on 
interval jobs. The approximation factor is tight. 
\end{theorem}

\subsection{A $3$-approximation for busy time 
with non-interval non-preemptive jobs}

The busy time problem for flexible jobs
was studied by Khandekar et al.~\cite{FSTTCS10}\footnote{In
their paper, the problem is called real-time scheduling.}.
They gave a $5$-approximation for this problem 
when the interval jobs can have arbitrary widths. 
For the unit width interval job case, their 
analysis can be modified to give a $4$-approximation. 
As a first step towards proving the $5$-approximation for 
flexible jobs of non-unit width, 
Khandekar et al.~\cite{FSTTCS10} prove 
that if $g$ is unbounded, then the
problem is polynomial-time solvable. 
The output of their dynamic program converts 
an instance of jobs with flexible windows 
to an instance of interval jobs,
by fixing the start and end times of every job. 

\begin{theorem}
\label{thm:khandekar}		
\cite{FSTTCS10} If $g$ is unbounded, 
the busy time scheduling problem is 
polynomial-time solvable. 
\end{theorem}

From Theorem~\ref{thm:khandekar}, the busy time of 
the output of the dynamic program 
on the set of (not necessarily interval) 
jobs $\textscr{J'}$ is equal to $OPT_\infty(\textscr{J'})$. 

Once Khandekar et al.~\cite{FSTTCS10} obtain 
the modified interval instance, they apply their 
$5$-approximation to interval jobs of arbitrary
widths to get the final bound. 
However, for jobs having unit width, their algorithm 
and analysis can be modified to apply the 
$4$-approximation algorithm of Flammini 
et al.~\cite{IPDPS09} for interval 
jobs with bounded $g$ to get a final bound of four. 
Moreover, extending the algorithms of Alicherry and 
Bhatia~\cite{ESA03} 
and Kumar and Rudra~\cite{FSTTCS05} to the general busy time problem, 
by converting an instance of flexible jobs 
to an interval job instance 
(similar to Khandekar et al.~\cite{FSTTCS10}) 
also gives a $4$-approximation\footnote{The bound of four for these 
algorithms is tight, as shown in the Appendix~\ref{app:flexible}.}. 

We give a $3$-approximation for the busy time problem,
improving the existing $4$-approximation. 
Analogous to Khandekar et al.~\cite{FSTTCS10}, 
we first convert the instance $\textscr{J}'$ to an instance 
$\textscr{J}$ of interval jobs by running 
the dynamic program on $\textscr{J}'$ and then
fixing the job windows according to the start times
found by the dynamic program.  
Then we run the \textsc{GreedyTracking} algorithm 
described below on $\scr{J}$.
For the rest of this section, we work with the instance
$\scr{J}$ of interval jobs.
%
Before describing the algorithm, consider the notion of a 
\emph{track} of jobs. 

\begin{definition}
A \emph{track} is a set of interval jobs with 
pair-wise disjoint windows. 
\end{definition}

Given a feasible solution, one can think of each
bundle $\textscr{B}$ as the union of 
$g$ individual tracks of jobs. 
The main idea behind the algorithm is to 
identify such tracks iteratively, bundling the first $g$ 
tracks into a single bundle, the second $g$ 
tracks into the second bundle, etc.
\textsc{FirstFit} \cite{IPDPS09} suffers from the fact that 
it greedily considers jobs one-by-one; 
\textsc{GreedyTracking} is less myopic in that 
it identifies jobs whole tracks at a time. 

In the $i^{th}$ iteration, 
\textsc{GreedyTracking} identifies a track 
$\textscr{T}_i \subseteq \textscr{J} \setminus \bigcup_{k=1}^{i-1} 
\textscr{T}_k$ 
of maximum length $\ell(\textscr{T}_i)$ and assigns it 
to bundle $B_p$, where $p= \left\lceil\frac{i}{g}\right\rceil$. 
One can find such a track efficiently 
by considering the job lengths as their weights and finding the 
maximum weight set of interval jobs with disjoint windows via
weighted interval scheduling algorithms~\cite{CLRS}.
Denoting by $\kappa$ the final number of bundles,
\textsc{GreedyTracking}'s total busy time is 
$\sum_{i=1}^\kappa Sp(\textscr{B}_i)$.
The pseudocode for \textsc{GreedyTracking} is provided in 
Algorithm~\ref{alg:Tracking}.

\begin{algorithm}
\begin{algorithmic}[1]
\caption{\textsc{GreedyTracking}. Inputs: $\textscr{J}$, $g$.}
\label{alg:Tracking}
\STATE $\textscr{S} \leftarrow \textscr{J}$, $i\leftarrow 1$.
\WHILE{$\textscr{S} \neq \emptyset$}
	 \STATE Compute the longest track $\textscr{T}_i$ from $\textscr{S}$ and assign it to bundle $B_{\lceil\frac{i}{g}\rceil}$. 
	 \STATE $\textscr{S} \leftarrow \textscr{S}\setminus \textscr{T}_i$, $i\leftarrow i+1$.
\ENDWHILE
\STATE Return {bundles 
$\{\textscr{B}_p\}^{\lceil\frac{i-1}{g}\rceil}_{p=1}$}
\end{algorithmic}
\end{algorithm}

\begin{theorem}
\textsc{GreedyTracking} is 3-approximate.
\end{theorem}

\begin{proof}
%
By Observation~\ref{obs:spinfty}, 
$Sp(\scr{B}_1) \le OPT_\infty(\scr{J}') \le OPT(\scr{J}')$. 
Therefore, it suffices to show that
$\sum_{i>1} Sp(\scr{B}_i) \le \frac{2}{g} 
\ell(\scr{J}') \le 2OPT(\scr{J}')$.  
We will achieve this by charging the span of bundle $\scr{B}_i$
to the mass $\ell(\scr{B}_{i-1})$, for $i>1$.  
In particular, if we could identify
a subset $\scr{Q}_i$ of jobs in $\scr{B}_i$ 
with span $Sp(\scr{Q}_i) = Sp(\scr{B}_i)$ and
with the additional property that at most two jobs of $\scr{Q}_i$ are
live at any point in time, then

\[
	Sp(\scr{B}_i) = Sp(\scr{Q}_i) \le \ell(\scr{Q}_i) 
	\le 2\ell(\scr{T}^\star) \le \frac{2}{g} \ell(\scr{B}_{i-1})
\]
\noindent
where $\scr{T}^\star$ is the first track of bundle $\scr{B}_i$.
The right-most inequality follows by the greedy nature
of \textsc{GreedyTracking}, as does the inequality preceding
it: $\ell(\scr{T}^\star)$ is at least that of any other track in
$\scr{B}_i$ and at most the average over tracks in $\scr{B}_{i-1}$.

To find $\scr{Q}_i$, start with jobs of $\scr{B}_i$
and remove any job $J_j$ 
whose window is a subset of another job $J_i$'s window, i.e.
such that $[r_j,d_j) \subseteq [r_i, d_i)$.  This can be
done in polynomial time.  The subset $\scr{Q}'_i$ of remaining jobs 
has the property that for any two jobs $J_j$ and $J_i$ in $\scr{Q}'_i$,
if $r_j < r_i$, then $d_j \le d_i$.  As in the literature,
instances with this structure are 
called ``proper'' instances~\cite{IPDPS09}.
Sort jobs of $\scr{Q}'_i$ in non-decreasing order by release time.
Iteratively add to subset $\scr{Q}_i$ from these jobs, breaking ties 
in favor of jobs later in the ordering.  Initially, $\scr{Q}_i$ is empty.
Repeat the following until $\scr{Q}'_i$ is empty:
let $d_\mathbf{max}$ be the current maximum deadline of $\scr{Q}_i$, 
or 0 if $\scr{Q}_i$ is empty.  Consider the jobs 
in $\scr{Q}'_i$ that are live at $d_\mathbf{max}$.
Remove from $\scr{Q}'_i$ all but the ``last'' one (i.e., the one with
latest deadline); move this ``last'' one from $\scr{Q}'_i$ to $\scr{Q}_i$.
When this process terminates, $\scr{Q}_i$ 
will have the two properties we want.
Suppose that three jobs $J_1, J_2, J_3$
of $\scr{Q}_i$ were live at $t$ with $r_1 \le r_2 \le r_3$.
Then at the time the process added $J_1$ to $\scr{Q}_i$, 
it considered $J_3$ as a possible ``last'' job, and
$J_2$ could never have been added to $\scr{Q}_i$.
So, no more than two jobs of $\scr{Q}_i$ are live at any point in time.
Also, by construction, $Sp(\scr{Q}_i) = Sp(\scr{B}_i)$.
\qed
\end{proof}

Figure~\ref{fig:factor3} shows that the approximation factor of $3$ achieved by
\textsc{GreedyTracking} is tight. 
In the instance shown, a gadget of $2g$ interval jobs is repeated $g$ times. In this gadget, 
there are $g$ identical unit length interval jobs which overlap for $\epsilon$ amount with 
another $g$ identical unit length interval jobs. The $g$ gadgets are disjoint 
from one another, which means, there is no overlap among the jobs of any two gadgets. 
There are $2g$ flexible jobs, whose windows span the windows of all the $g$ gadgets. 
These jobs are of length $1-\frac{\epsilon}{2}$. An optimal packing would pack each set of $g$ identical 
jobs of each gadget in one bundle, and the flexible jobs in 2 bundles, giving a total busy time 
of $2g+2-\epsilon$. However, the dynamic program minimizing the span does not take capacity 
into consideration, hence in a possible output, the flexible jobs may be packed $2$ 
each with each of the $g$ gadgets, in a manner such that they intersect with all of the jobs of the gadget. 
Hence, the flexible jobs cannot be considered in the same track as any unit interval job 
in the gadget it is packed with. Due to the greedy nature of \textsc{GreedyTracking}, 
the tracks selected would not consider the flexible jobs in the beginning, and the interval jobs may also 
get split up as in Figure \ref{fig:factor3-1}, giving a total busy time of 
$4(1-\epsilon)g + (2 - o(\epsilon))g = (6-o(\epsilon))g$, hence it approaches a factor 
$3$ asymptotically.

\begin{figure}[htbp]
	\centering
		\includegraphics[width=0.75\textwidth]{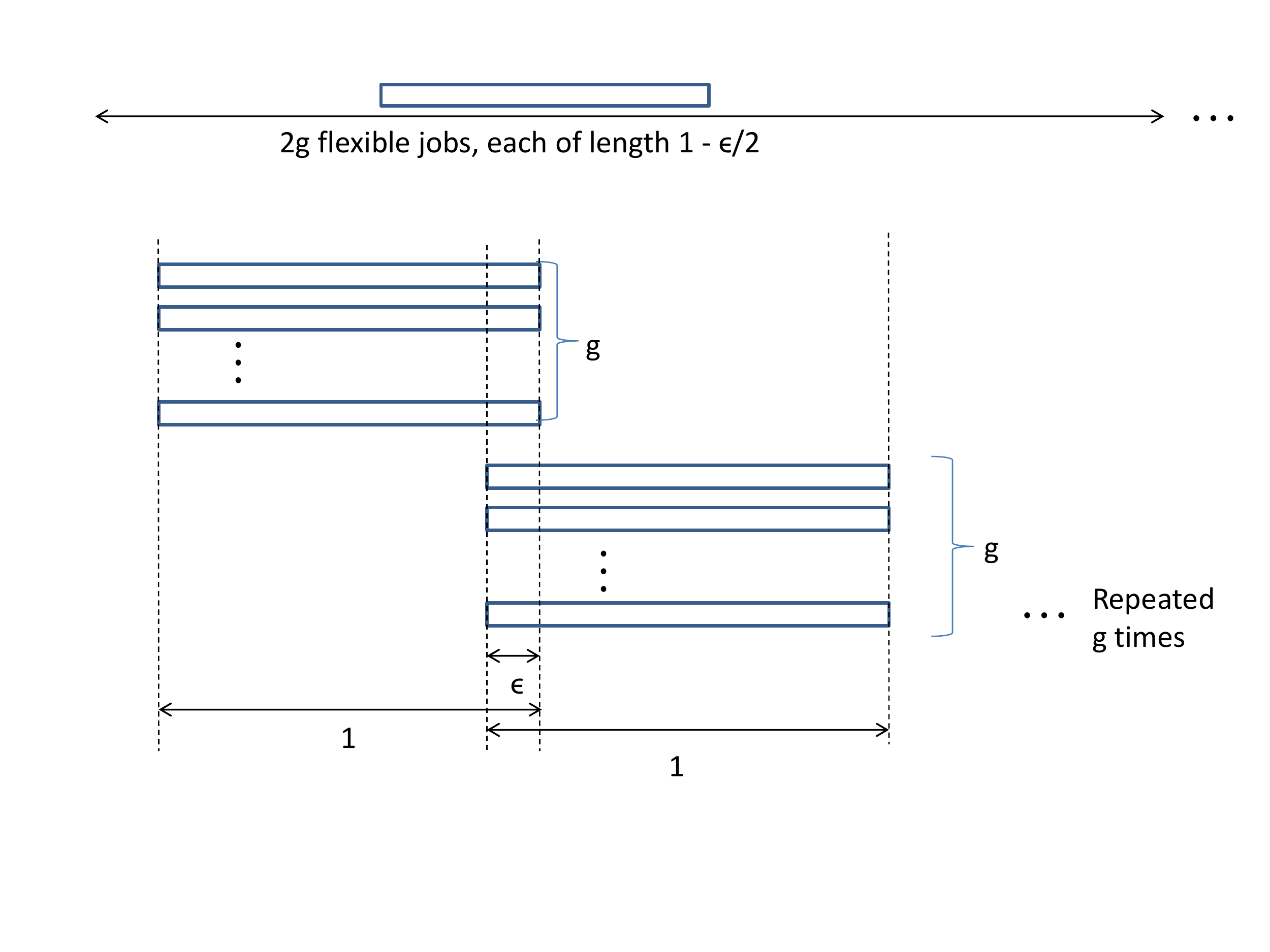}
	\caption{Gadget for factor 3 for \textsc{GreedyTracking}}
	\label{fig:factor3}
\end{figure}

\begin{figure}
	\centering
		\includegraphics[width=0.75\textwidth]{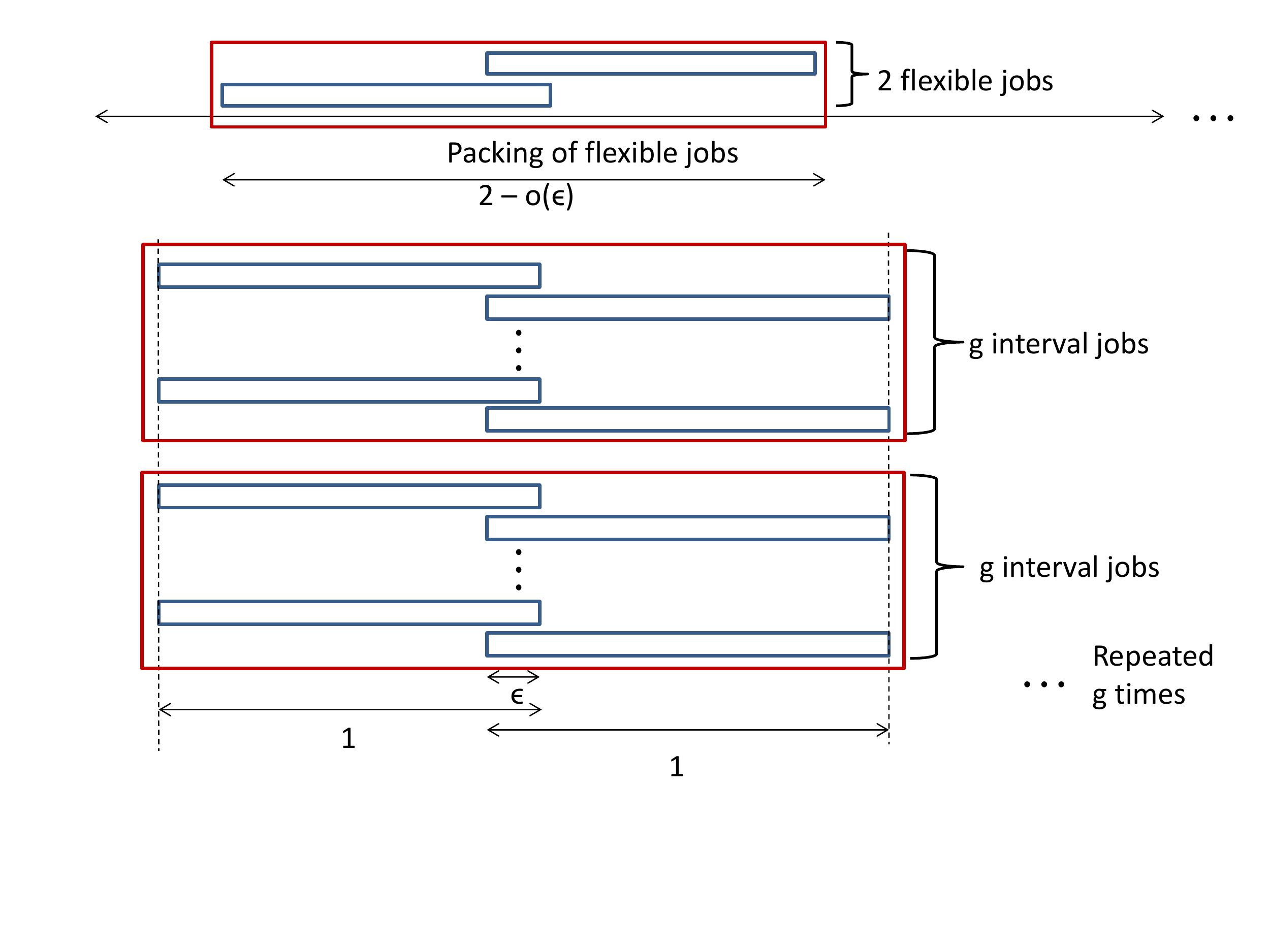}
	\caption{Possible packing by \textsc{GreedyTracking}}
	\label{fig:factor3-1}
\end{figure}


\subsection{Preemptive busy time}
In this section, we remove the restriction, a job needs 
to be assigned to a single machine. 
A job $j$ needs to be assigned a total of $p_j$ time units 
within the interval 
$[r_j, d_j)$ and at most one machine may 
be working on it at any given time.

\begin{theorem}
For unbounded $g$ and preemptive jobs, 
there is an exact algorithm to minimize busy time.
\end{theorem}	

\begin{proof}
The algorithm is a simple greedy one. 
Let $\mathcal{J}_1$ be the set of jobs 
of earliest deadline $d_1$ and let the longest 
job $j_{\max,1}$ in $\textscr{J}_1$ have length $\ell_{\max,1}$. We open 
the interval $[d_1-\ell_{\max,1}, d_1)$, and for every the job $j \in \textscr{J}$ 
such that $[r_j, d_j)\cap[d_1-\ell_{\max,1}, d_1) \neq \emptyset$, we schedule 
it up to $d_1 - r_j$ in the interval $[r_j, d_1)$. Then we shrink the interval 
$[d_1-\ell_{\max,1}, d_1)$ and adjust the windows and 
remaining processing lengths of the jobs in $\textscr{J}$ 
and then repeat till all jobs in $\textscr{J}$ have been completely scheduled. 

In the first iteration, without loss of generality, 
the optimal solution will also 
open the interval $[d_1-\ell_{\max,1}, d_1)$; 
$j_{\max,1}$ has to be scheduled completely $d_1$ and 
since $d_1$ 
is the earliest deadline, opening this length of interval 
as late as possible ensures that we can schedule the maximum length of 
any job in the instance $\textscr{J}$ with $j_{\max,1}$. 
The correctness follows by induction on the remaining iterations. 
\qed
\end{proof}

As a consequence, one can approximate
preemptive busy time scheduling for bounded $g$.  
First, solve the instance under the assumption 
that $g$ is unbounded; denote by $\textscr{S}_\infty$
this (possibly infeasible) solution.
The busy time of $\textscr{S}_\infty$ is $OPT_\infty(\textscr{J})$,
and is a lower bound on the optimal solution for bounded $g$.
The algorithm for bounded $g$ 
will commit to working on job $j$ precisely 
in the time intervals where $\textscr{S}_\infty$ had scheduled it.  
Partition the busy time of $\textscr{S}_\infty$
into the set of interesting intervals $\{I_1, \ldots, I_k\}$, 
where $k = \theta(n)$. 
 
For every interesting interval $I_i$, assign the jobs 
scheduled in $I_i$ to 
$\lceil \frac{n(I_i)}{g} \rceil$ machines in arbitrary order, 
filling the machines  
greedily such that there is at most one machine 
with strictly less than $g$ jobs. 

For each $I_i$, at most one machine contains less than $g$ jobs, which 
we charge to $OPT_\infty(\textscr{J})$ 
All other machines are at capacity, i.e., 
have exactly $g$ jobs and hence we charge them to $\frac{\ell(\textscr{J})}{g}$. 
This implies an approximation of $2$. 

\begin{theorem}
There is a preemptive algorithm
whose busy time is at most twice that
of the optimal preemptive solution, for bounded $g$.
\end{theorem}

\bibliographystyle{plain}
\bibliography{references}

\appendix
\section{Appendix of full version: Interval Jobs}
\label{app:intjobs}

\subsection{Kumar and Rudra's Algorithm} 
\label{sec:kr}
Here we provide an overview of the algorithm of Kumar and Rudra 
\cite{FSTTCS05} 
for fiber minimization problem which implies a $2$-approximation for 
the busy time problem on interval jobs. 
The fiber minimization problem is as follows. An optical fiber network 
needs to 
satisfy the given set of requests, that need to be 
assigned to consecutive links or edges connected in a line. 
There are $n$ of these links. Each request needs some 
links $\{i, i+1, \ldots, j\}$, where $1\leq i<j\leq n$. 
Each segment of an optical fiber can support 
$\mu$ wavelengths over the consecutive links that it spans, 
and no two requests can be assigned the same wavelength on 
the same fiber, if they need to use the same link.  
We want a feasible assignment of the requests such that total 
length of optical fiber used is minimized. 
Notice, this is very similar to the busy time problem on interval jobs. 
Think of the requests as interval jobs. If a request needs the consecutive 
links 
$\{i, i+1, \ldots, j\}$, where $1\leq i<j\leq n$, then this can be 
equivalently 
thought of as an interval job with release time $i$ and deadline $j$, 
i.e., with a window $[i,j)$, with processing length $j-i$, 
in a discrete setting where time is slotted. The total number of links 
being $n$, the processing lengths 
of the jobs here is linear. In this case, we can think of each slot as an 
interesting 
interval (since jobs begin and end only at slots) and define the demand 
profile as the tuples $(i,D(i))$, 
where $i$ is a time slot in $\{1,\ldots, n\}$. 
Their algorithm proceeds in two phases. In the first phase they assign the 
jobs 
to levels within the demand profile (where the total number of levels 
equals  
the maximum raw demand at any point), and potentially 
allow for a limited infeasibility in this packing. Specifically, at most 
two jobs 
can be assigned to the same level anywhere within the demand profile. In 
the second phase, they give a feasible 
packing of the jobs, considering $\mu$ levels at a time, 
removing the infeasibility introduced earlier, but without 
exceeding the cost by more than a factor of $2$. This is done as follows. 
For levels $\{(i-1)\mu+1, \ldots, i\mu\}$, ($i \in \{1, \ldots, 
D_{max}\})$ where 
$D_{max}$ is the maximum height of the demand profile), 
Kumar and Rudra \cite{FSTTCS05} open two fibers, instead of one, and 
assign 
jobs to fibers, such that two jobs which were assigned to the same 
level in the demand profile, get assigned to separate fibers according to 
a 
simple parity based assignment. Their analysis assumes that  
the raw demand at every time slot $t$, $|A(t)|$ is a multiple of $\mu$ and 
charges 
to such a demand profile. It is clear that the demand profile gets charged 
at most twice, respecting the $\mu$ capacity constraint of the fibers, and 
since the demand profile is a lower bound on the cost of an optimal solution, 
this gives a $2$-approximation algorithm. 

The polynomial time complexity of the algorithm 
crucially depends on the fact that we have $n$ links, and hence the job 
lengths being linear, we need to consider only a linear number of slots. 
The above does not hold for the busy time problem for interval jobs with 
arbitrary release 
times, deadlines and processing lengths. The number of time instants to 
consider 
may not be polynomial. 
However, the key observation is that even if the release times and 
deadlines 
of jobs are not integral, there can be at most $2n$ \emph{interesting} 
intervals, 
such that no jobs begin or end within the interval. The demand profile 
is uniform over every interesting interval. Therefore, their algorithm 
can be applied to the busy time problem, with this simple modification, 
still maintaining the polynomial complexity. The assumption 
regarding multiple of $\mu$ (in the busy time case, this would be $g$) 
at every slot, would translate as a multiple 
of $g$ jobs over every interesting interval. However, note that for an 
arbitrary instance, 
we can add dummy jobs spanning any interesting interval $I_i$ where the 
raw demand 
$|A(I_i)|$ is not a multiple of $g$ without changing the demand profile. 
Specifically, if $cg< |A(I_i)|<(c+1)g$, for some $c\geq 0$, then $DeP(I_i) 
= c+1$, 
hence adding $(c+1)g - |A(I_i)|$ jobs spanning $I_i$  
does not change the demand profile. Thus we can apply their algorithm on 
the busy time 
instance, where the demand profile is defined only interesting intervals 
and 
the demand everywhere is a multiple of $g$. The assignments to the fibers 
as done by their algorithm in Phase 2, will give the bundles for the busy 
time 
problem. 

\subsection{Alicherry and Bhatia's Algorithm}
\label{sec:ab}
Now, we describe how the work of Alicherry and Bhatia \cite{ESA03}, 
implies 
another, 
elegant algorithm with a $2$-approximation for interval jobs. 
Alicherry and Bhatia study a generalized coloring and routing problem 
on interval and circular graphs, motivated by optical design systems. 
Though the problems they consider are not directly related to the busy 
time 
problem, we can use their techniques to develop the $2$-approximation 
algorithm. 
Similar to Kumar and Rudra's work, their goal is to route certain 
requests, 
which require to be assigned to consecutive links or edges in the interval 
or circular graph. At each link, we color the requests assigned to that 
link. 
The colors are partitioned into sets, which are ordered, such that colors 
in the 
higher numbered sets cost more. The total cost of the solution is the sum 
of the costs of the 
highest colors used at all the links, and the objective is to minimize 
this cost.  
Though this problem seems quite different 
from the busy time problem on interval jobs, the one of the key 
observations is that the cost 
needs to be a monotonically non-decreasing function respecting the set 
order. 
It need not be a strictly increasing function. Hence, we can think of the 
sets numbered 
in a linear order, and give each set $g$ colors. We set the number of all 
the colors 
in a set $i$ as $i$. If $c\cdot g + k$ requests use a link, the cost of 
that 
link would be the cost of the highest color used at the link, which is 
$c+1$. 
Hence, what we are really summing is the total cost of the 
demand profile defined on a set of interval jobs, which have integral 
release times, and deadlines, and linear processing lengths, since the 
number of 
links is $n$ (part of the input). Therefore, a $2$ approximation algorithm 
minimizing the cost is really providing a solution that costs at most 
twice 
the demand profile of this restricted instance. The technique used 
involves setting up 
a flow graph with a certain structure, depending on the current demands or 
requests as yet unassigned. 
It can be easily proved that the graph has a cut of size at least $2$ 
everywhere if the demand 
everywhere is at least $2$. Now, we find a flow of size two in this graph 
from the source to the sink. Each flow path will consist of a set of 
disjoint requests or demands 
(where the disjointness refers to the links they need to use), 
and the union of the two flows will reduce a demand of at least 
unity from every link. This is repeated till the demand is $0$ or $1$ 
everywhere. 

As in Section \ref{sec:kr}, we use the following observation: 
the time slots can be considered to be interesting 
intervals for a set of interval jobs. The busy time instance 
with non-polynomial job lengths and arbitrary release times and deadlines 
has a linear number of interesting intervals, and 
hence we can think of our instance in this discretized setting. Therefore, 
we can 
apply their algorithm, modified accordingly, to our problem to get a solution of cost 
within twice of the optimal solution. The algorithm will consider a busy 
time 
instance with the demand profile defined on interesting intervals and with 
a multiple of $g$ jobs everywhere without any loss of generality. 
It will first open up two bundles. The flow graph is then set up as 
defined by Alicherry and Bhatia. 
For the first $g$ 
iterations, the algorithm will find $2g$ flow paths (each consisting of 
disjoint interval jobs), the union of which removes at least 
a demand of $g$ from everywhere. We assign $g$ of these paths to one 
bundle 
and the remaining $g$ to the other. Each flow path consists of disjoint 
jobs, 
hence, each bundle 
will have at most $g$ jobs at time instant. Moreover, together, 
these bundles have removed a demand $g$ from everywhere in the demand 
profile,  
hence they have charged the lowermost level (which is also the widest 
level) of the demand 
profile at most twice. 
The demand profile is now suitably modified after removing the jobs 
already assigned. 
Once again two bundles are opened, and the same procedure is performed for 
the next $g$ 
iterations. This continues till the demand profile becomes empty 
everywhere, in other 
words, all jobs are assigned. The resultant bundles are feasible and 
charge the demand 
profile at most twice. 

\subsection{Lower bound}
\label{sec:lb}
Though the upper bound of $2$ was shown by Kumar and Rudra \cite{FSTTCS05} 
and 
Alicherry and Bhatia \cite{ESA03} 
for their algorithms, a lower bound on the performance of the algorithms 
was not provided. 
Here we show that for both these algorithms, the approximation ratio 
obtained can be arbitrarily close to $2$. 
Figure \ref{fig:example-ab-kr} shows an instance of interval jobs, for 
which both the algorithms implied by the work of Kumar and Rudra and 
Alicherry and Bhatia approach a factor of $2$ of the optimal solution. 
In this example, $g=2$ and there are two interval jobs of length $1$, one 
interval 
job of length $\epsilon$, one of length $\epsilon'<\epsilon$, and one of 
length $\epsilon - \epsilon'$. 
As required by the analysis of Kumar and Rudra and Alicherry and Bhatia, 
the demand everywhere is a multiple of $g$. 
A possible output by both algorithms (adapted to the busy time problem as 
described) 
has cost $2+\epsilon$, whereas the optimal solution has cost $1+\epsilon$. 
For $\epsilon \rightarrow 0$, the approximation factor approaches $2$. 

\begin{figure}[htbp]
	\centering
		\includegraphics[width=0.75\textwidth]{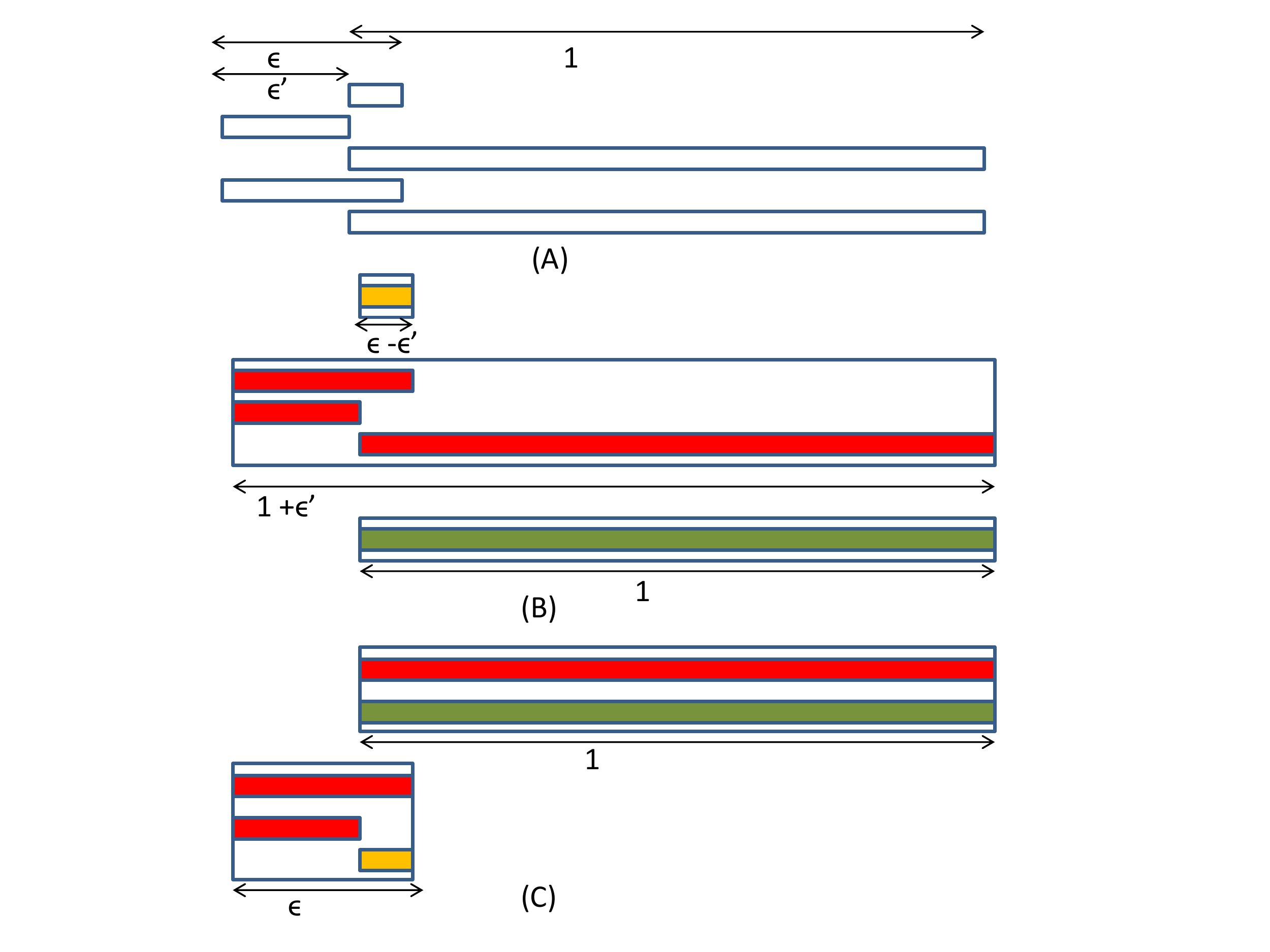}
	\caption{(A) An instance of interval jobs and $g=2$. (B) A 
possible output by the algorithms of Kumar and Rudra \cite{FSTTCS05} and 
Alicherry and Bhatia \cite{ESA03}, of cost = $2+\epsilon$. (C) The optimal 
solution of cost $1+\epsilon$.}
	\label{fig:example-ab-kr}
\end{figure}

\begin{theorem}
\label{thm:busy2}
There exist $2$-approximation polynomial time algorithms 
for the busy time problem on interval jobs. 
The approximation factor is tight. 
\end{theorem}

\begin{proof}
The proof follows from the discussions of Sections \ref{sec:kr}, 
\ref{sec:ab}, and \ref{sec:lb}.
\qed
\end{proof}

\section{Appendix of full version: Flexible Jobs}
\label{app:flexible}
\subsection{Prior $4$-approximation}
In this section we discuss the busy time problem with flexible jobs. 
This problem was studied by Khandekar et al.~\cite{FSTTCS10}, who 
refer to this problem as the real-time scheduling problem. 
They gave a $5$-approximation for this problem when the jobs can 
have arbitrary widths. For the unit width jobs, their 
analysis can be modified to give a $4$-approximation. 

As a first step towards proving the $5$-approximation for 
flexible jobs of non-unit width, Khandekar et 
al.~\cite{FSTTCS10} prove 
that if $g$ is unbounded, then this 
problem is polynomial-time solvable. 
The output of their dynamic program essentially converts 
an instance of jobs with flexible windows 
to an instance of interval jobs 
(with rigid windows), 
by fixing the start and end times of every job. 

\begin{theorem}
\label{thm:khandekar_app}
~\cite{FSTTCS10} If $g$ is unbounded, the real-time scheduling problem is 
polynomial-time solvable. 
\end{theorem}

From Theorem \ref{thm:khandekar_app}, 
the busy time of the output of the dynamic program 
on the set of (not necessarily interval) 
jobs $\textscr{J}$ is equal to 
$OPT_\infty(\textscr{J})$. 

Once Khandekar et al. obtain the modified interval instance, they 
apply their $5$-approximation algorithm for non-unit width interval 
jobs to get the final bound. 
However, for jobs 
with unit width, one can apply the same dynamic program to convert the instance 
to interval jobs and then
apply the $4$-approximation algorithm of Flammini et al. \cite{IPDPS09} 
for interval 
jobs with bounded $g$ to get the final bound of $4$. 

The $2$-approximation algorithm \cite{FSTTCS10} for interval instance 
charges the demand profile, hence it is immediately not clear how to 
extend 
it to handle flexible jobs since the demand profile cannot be 
defined analogous to the interval case. One possible natural extension 
is to follow the approach of Khandekar et al., to convert a flexible 
instance 
to an interval instance, and then apply the algorithm to this modified 
instance. 
Furthermore, the algorithm of Kumar and Rudra assumes that the demand 
profile everywhere is a multiple 
of $g$. Hence, after modifying the instance to an interval instance, 
we need to add dummy jobs accordingly to interesting intervals to bring up 
their demands to multiples of $g$.  
However, there exists an instance where this algorithm will approach 
a factor of $4$ of the optimal solution. 
This is the worst that it can do, since we prove in the following lemma 
that the demand profile of the modified instance of interval jobs 
is at most twice the demand profile of the optimal solution (note that 
once 
the jobs have been assigned in the optimal solution, their positions get 
fixed, and hence the demand profile can now be computed easily). 

\begin{lemma}
\label{lemma:dp2}
The demand profile of the output of the dynamic program converting 
the flexible jobs to interval jobs is at most $2$ times the demand profile 
of an optimal solution structure. 
\end{lemma}
\begin{proof}
The objective function of the dynamic program (Theorem \ref{thm:khandekar_app}) is to 
minimize the total 
busy time of a flexible job instance assuming $g$ is unbounded. Since the 
dynamic program 
is optimal, it will pack as many jobs and as much length as possible 
together. Hence, 
if a job has a choice of being assigned to a spot where other jobs need to 
be assigned 
as well, then it will be assigned at that spot instead of at some other 
spot where 
no jobs need to be assigned. Therefore, at any level of the demand 
profile, we 
can charge it to the mass of the level below, and if it is the first (or, lowest) 
level, we 
charge it to $OPT_\infty$ bound. Hence, in total the optimal solution gets 
charged twice, 
once by the mass bound, and once by the span bound, 
giving a $2$-approximation. 
\qed
\end{proof}

There exists an instance of flexible jobs for which the demand profile 
output by the dynamic program of Khandekar et al. approaches $2$ times the 
cost of the demand profile of the optimal solution structure. We have 
shown such an 
instance in Figure \ref{fig:busydp}. The instance consists of the 
following 
types of jobs: one interval job of unit length, followed by $(g-1)$ 
disjoint sets of identical $g$ interval 
jobs, where in the $i^{th}$ set, each job is of length $1+i\epsilon$, ($i 
\in \{1, \ldots, (g-1)\}$). 
Apart from these, there are $g-1$ flexible jobs, where the $i^{th}$ job is 
of  
length $1+i\epsilon$, where $i\in \{1, \ldots, (g-1)\}$ and has a feasible 
window spanning the 
the windows of the first $i+1$ disjoint sets of interval jobs,  
as shown in the figure. 
An optimal solution would pack the $g-1$ flexible jobs with the first 
interval job, 
and the remaining $(g-1)$ disjoint sets of identical $g$ interval jobs in 
their respective windows, 
with a total busy time of $g + \left(\frac{g(g+1)}{2}-1\right)\epsilon$. 
The dynamic program however disregards capacity constraints of the 
machines, and simply 
tries to minimize the span of the solution. Hence, with a little effort it 
can be 
seen that the unique output of the dynamic program 
(as shown in Figure~\ref{fig:busydp})
would have a span of $g + \frac{g(g-1)}{2}\epsilon$, and the demand 
profile 
on imposing a capacity of $g$ is of cost $2g - 1 + g(g-1)\epsilon$, which 
approaches $2$ the cost of 
the optimal solution when $\epsilon \rightarrow 0$.

\begin{figure}[htbp]
	\centering
		\includegraphics[width=0.75\textwidth]{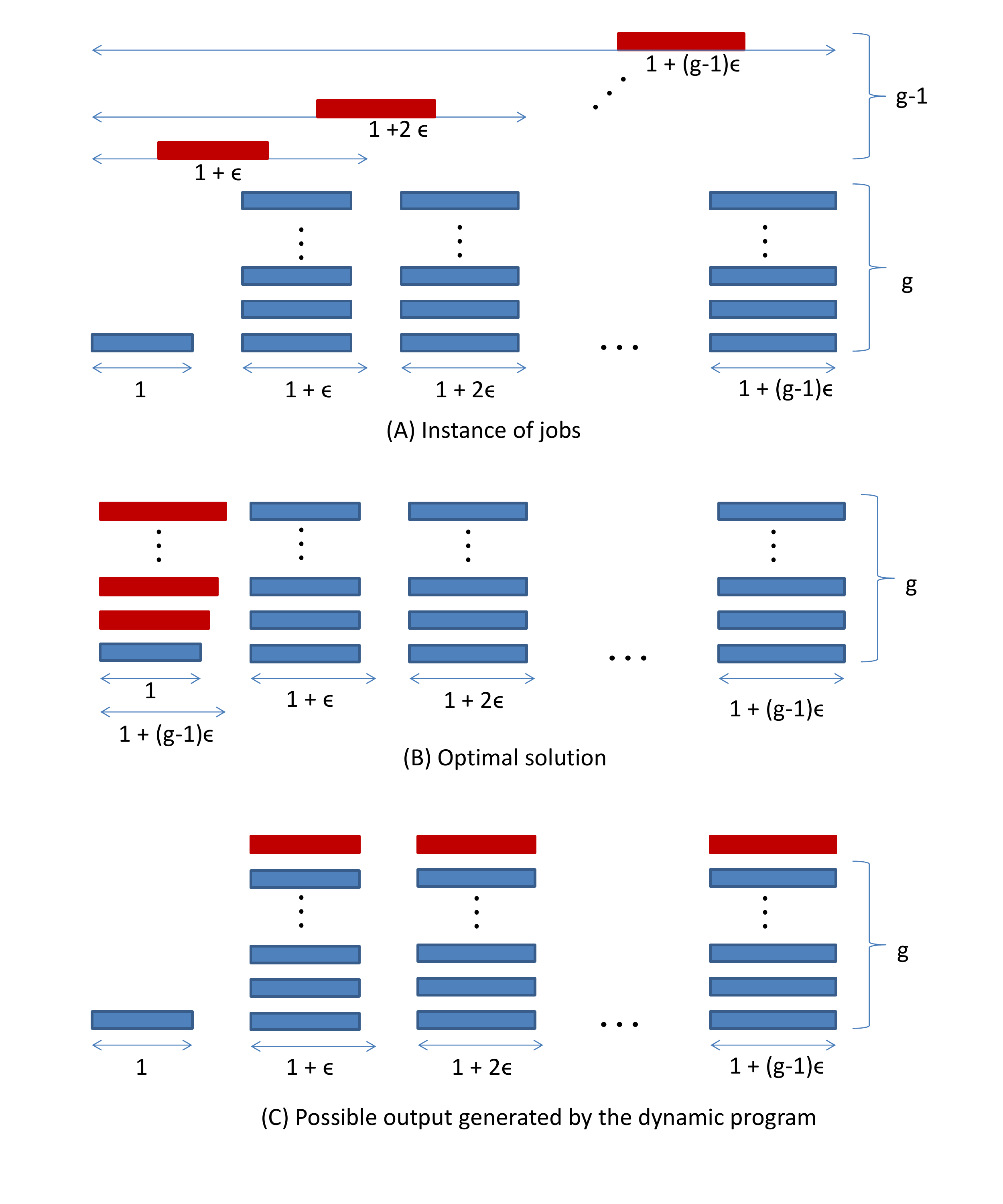}
	\caption{(A) An instance of interval and flexible jobs. (B)The 
optimal solution of busy time $g+ \frac{g^2 + g - 2}{2}\epsilon$. (C)The 
output of the dynamic program of Khandekar et al.~\cite{FSTTCS10} of busy 
time = $2g-1 +g(g-1)\epsilon$.}
	\vspace{-20pt}
	\label{fig:busydp}
\end{figure}

\begin{theorem}
A natural extension of the $2$-approximation algorithm of Kumar and Rudra 
\cite{FSTTCS05} 
(or the algorithm of Alicherry and Bhatia~\cite{ESA03}) 
for the interval jobs problem, to the flexible jobs problem, gives an 
approximation of $4$. This factor is tight. 
\end{theorem}
\begin{proof}
The approximation upper bound of $4$ follows from Lemma \ref{lemma:dp2} 
and Theorem \ref{thm:busy2}. 
However, there is a tight example as well. 
In this example, we have an instance of interval and flexible jobs. 
The instance consists of a unit length interval 
job, followed by $g-1$ disjoint occurrences of the gadget shown in Figure 
\ref{fig:factor4gadget}. 
The gadget consists of $g$ unit length interval jobs, $2g-2$ interval jobs 
of length $\epsilon$, 
$2$ interval jobs of length $\epsilon'$ and $2$ jobs of length $\epsilon - 
\epsilon'$, 
as shown in the figure. 
There are $g-1$ unit length flexible jobs, each with windows spanning the 
windows 
of the union of all of the interval jobs. 

\begin{figure}[htbp]
	\centering
		\includegraphics[width=0.75\textwidth]{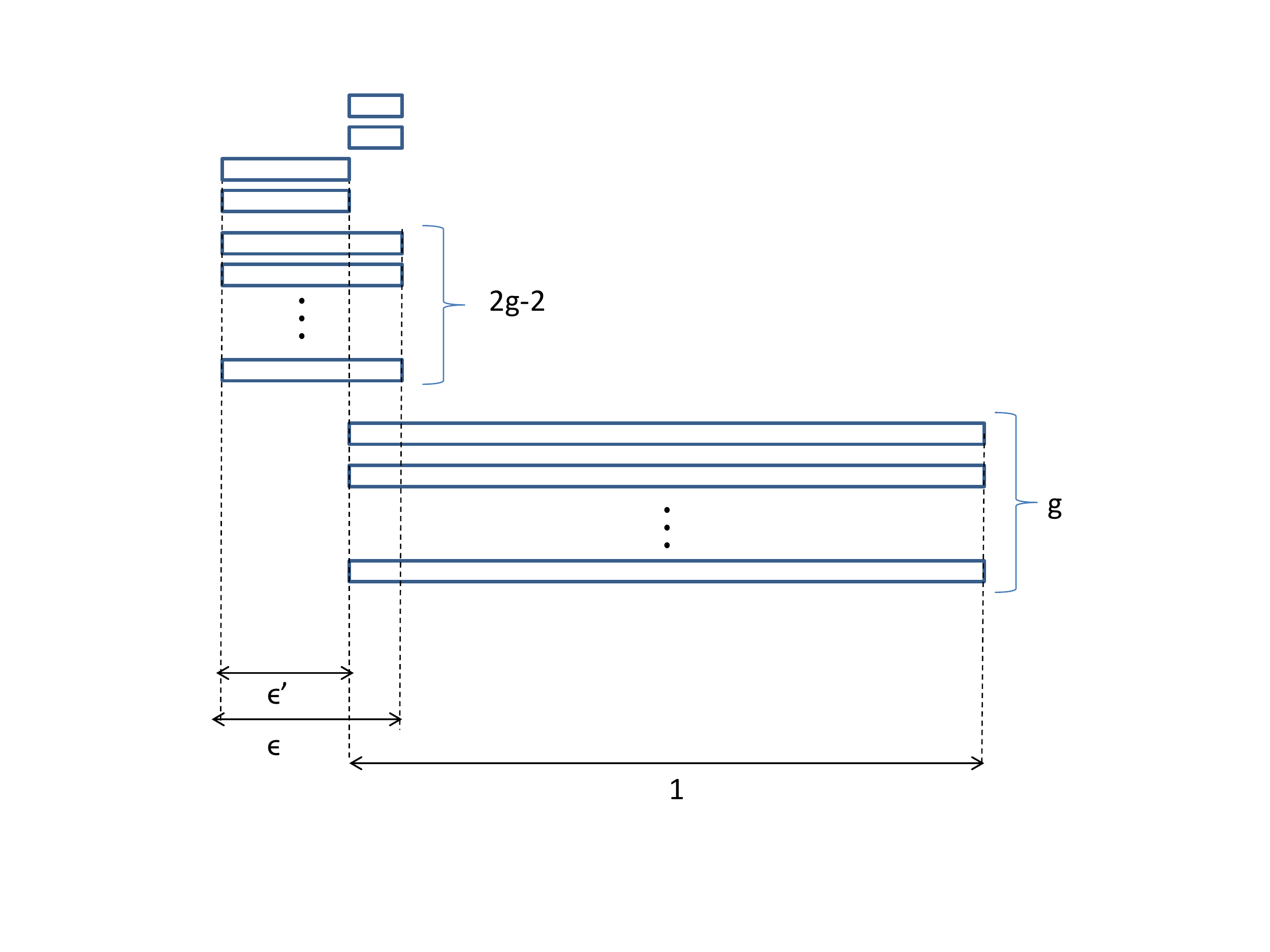}
	\caption{The gadget for the factor 4 example.}
	\label{fig:factor4gadget}
\end{figure}

On running the dynamic program to minimize span, a possible output is when 
each of 
$g-1$ flexible jobs are packed along with the $g-1$ gadgets. 
For applying the algorithms of Kumar and Rudra (or Alicherry and Bhatia), 
we need to make 
sure the demand everywhere is a multiple of $g$. Hence we add $g-1$ dummy 
jobs of unit length 
coincident with the first unit length interval job, as well as with each 
of the $g-1$ gadgets with 
a flexible job. This is shown in Figure \ref{fig:factor4dp}. 

\begin{figure}[htbp]
	\centering
		\includegraphics[width=0.90\textwidth]{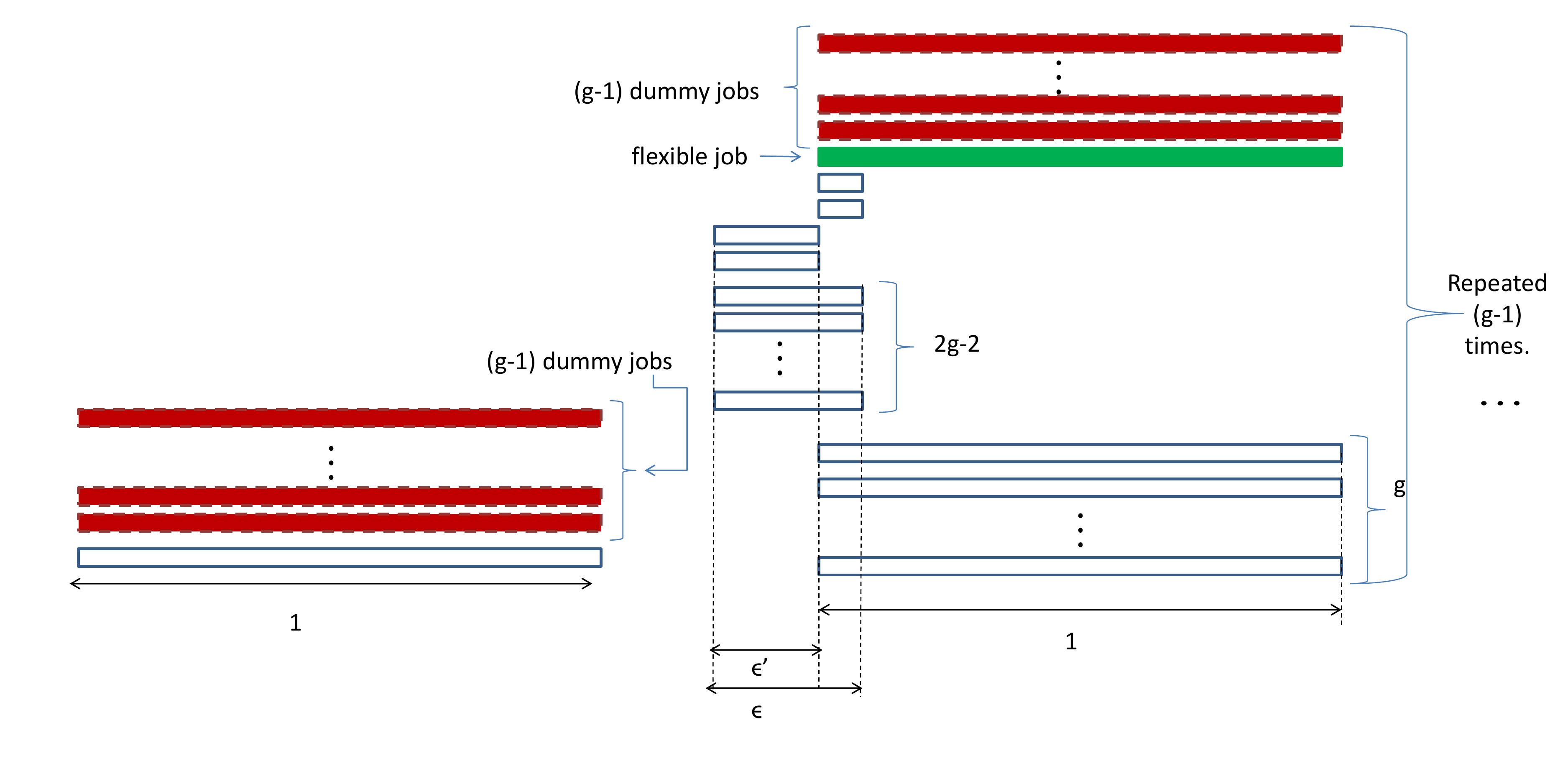}
	\caption{Output of the dynamic program on the instance of interval 
and flexible jobs for the factor 4 example.}
	\label{fig:factor4dp}
\end{figure}

Now, one possible run of the algorithm of Kumar and Rudra~\cite{FSTTCS05}
(or Alicherry and Bhatia~\cite{ESA03}) 
may result in the packing shown in 
Figure~\ref{fig:finaloutput4}, of cost $1 + 4(g-1) + O(\epsilon)$.
In contrast, the optimal solution packs the flexible jobs 
with the first unit-length interval job, 
and packs all the identical unit length jobs together, 
for a total cost of $g + O(\epsilon)$. 
Hence the ratio approaches $4$ for large $g$ and small $\epsilon$. 

\begin{figure}[htbp]
	\centering
		\includegraphics[width=0.70\textwidth]{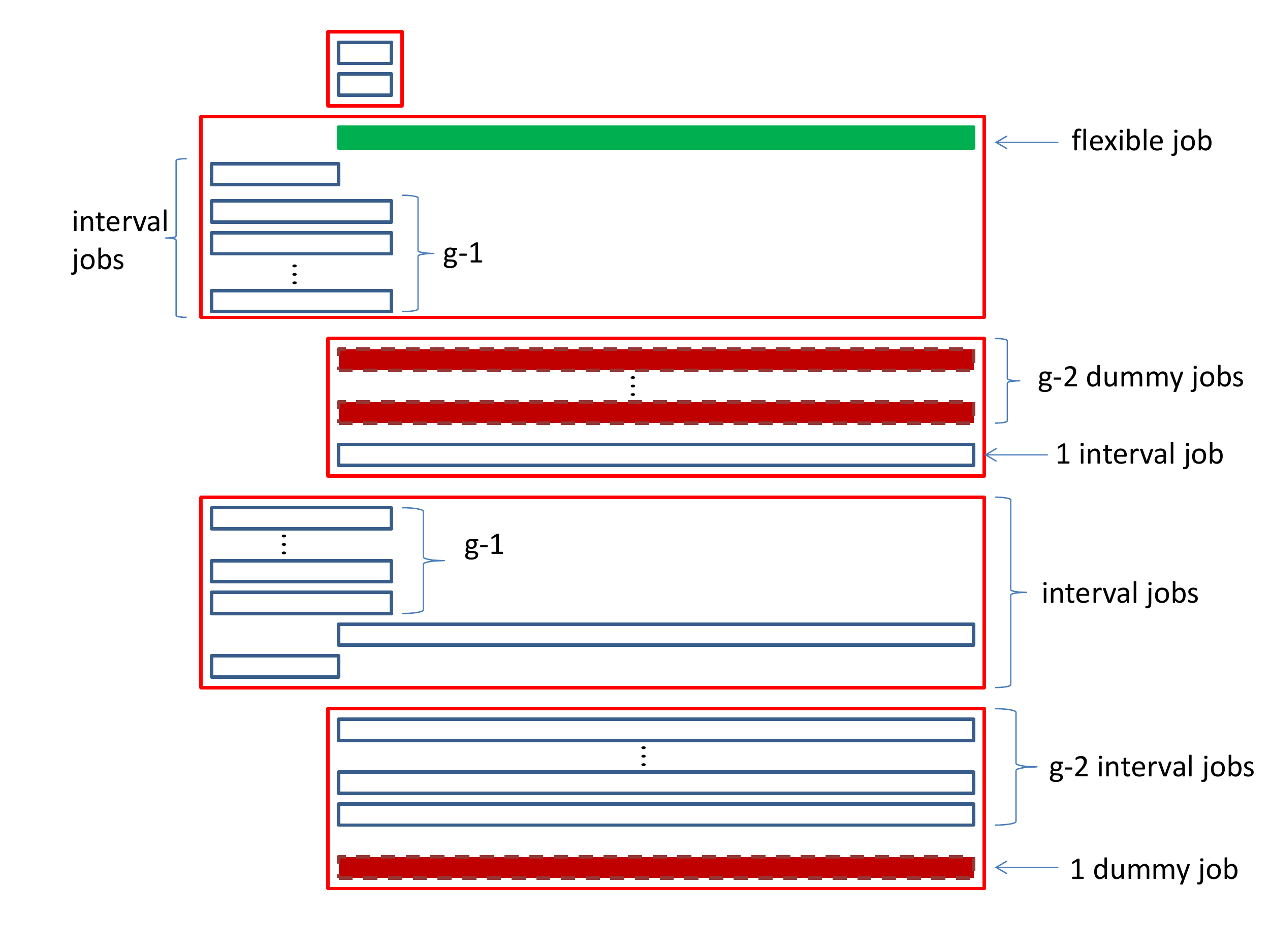}
	\caption{Possible busy time bundling produced by a 
run of Kumar and Rudra's~\cite{FSTTCS05} or 
Alicherry and Bhatia's~\cite{ESA03} algorithms on 
one gadget along with the flexible job and dummy jobs. }
	\label{fig:finaloutput4}
\end{figure}
\qed
\end{proof}

\end{document}